\documentclass[11pt,letterpaper]{article}

\usepackage[utf8]{inputenc}
\usepackage{authblk}

\usepackage{amssymb,amsmath,amsfonts}
\usepackage{amsthm}
\usepackage[margin=1truein]{geometry}

\usepackage{graphicx}
\usepackage[normalem]{ulem}

\usepackage{xcolor}
 \definecolor{bordeaux}{RGB}{100,0,50}
 \definecolor{darkblue}{RGB}{25, 25, 112}

\usepackage[pdftex,colorlinks=true,linkcolor=blue,citecolor=blue,urlcolor=red,unicode=true,hyperfootnotes=false,bookmarksnumbered]{hyperref}
\usepackage[capitalise, compress, nameinlink, noabbrev]{cleveref}

\newtheorem{ptheorem}{Theorem}
\newtheorem{plemma}[ptheorem]{Lemma}
\newtheorem{pcorollary}[ptheorem]{Corollary}
\newtheorem{pdefinition}[ptheorem]{Definition}

\AddToHook{env/plemma/begin}{\crefalias{ptheorem}{plemma}}

\crefname{ptheorem}{Theorem}{Theorems}
\crefname{plemma}{Lemma}{Lemmas}
\crefname{pdefinition}{Definition}{Definitions}

\hypersetup{linkcolor={{blue!70!black}}, citecolor={black}, urlcolor={blue!70!black}}
\hypersetup{linkcolor=bordeaux, citecolor = darkblue, urlcolor = darkblue}

\usepackage{enumerate,color}

\theoremstyle{plain}

\theoremstyle{definition} 

\newcommand{\N}{\mathbb{N}}
\newcommand{\R}{\mathbb{R}}

\author{
Bogumi{\l} Kami\'nski\thanks{Decision Analysis and Support Unit, SGH Warsaw School of Economics, Warsaw,
Poland; e-mail: bkamins@sgh.waw.pl}, 
Pawe{\l} Pra{\l}at\thanks{Department of Mathematics, Toronto Metropolitan University, Toronto, ON, Canada; e-mail: pralat@torontomu.ca}, 
Maria Sadza\thanks{High School Student, The National Mathematics and Science College, Coventry, United Kingdom; e-mail: maria.sadza@natmatsci.ac.uk}
}

\date{}

\title{Going in Circles:\ Collaborative Multi-Robot Treasure Hunting}

\begin{document}

\maketitle

\begin{abstract}
This paper investigates a multi-robot search-and-visit problem involving $n$ \emph{robots} starting at the origin and $k$ unknown \emph{treasures} hidden on the unit circle $\mathcal{C}$. The robots move independently anywhere in the plane with a maximum speed of 1 and instantly share discovery information. The goal is to determine an algorithm that minimizes the total time needed for all robots to individually visit every treasure. To establish the foundational limits of this task, we first solve an \emph{auxiliary} optimal routing problem for a single robot on $\mathcal{C}$ that possesses complete prior knowledge of the treasure locations. The exact analysis of this auxiliary problem yields explicit upper and lower bounds for the original search problem.
\end{abstract}

\section{Introduction and Main Results} 

Search algorithms have an extensive history of study in the mathematics and theoretical computer science literature. Examples of search problem variants and models include probabilistic search~\cite{stone1975}, game theoretic applications~\cite{alpern2002}, classical pursuit and evasion~\cite{nahin2012}, search problems and group testing~\cite{ahlswede1987}, etc. (For search problems on graphs, see~\cite{bonato2017}).

We investigate a variant of the search problem involving $n$ robots and $k$ treasures. Unlike standard search models that focus on the time required for searchers to find hidden targets, our goal is to minimize the total time required for every robot to individually reach every treasure. We can think of many real-world scenarios where this requirement naturally arises, from multi-disciplinary facility auditing, where specialized inspectors with unique expertise must each independently evaluate every detected defect, to hazardous site verification, where safety protocols demand that multiple distinct diagnostic robots independently inspect and clear every identified threat. In all these cases, it is not sufficient for the group of searchers to simply find targets; they must all reach these targets in order to achieve the goal of the problem.

Although discovering treasures is a central component of the problem, the cost of a solution is measured not by when treasures are found, but by the time required for every robot to physically visit every treasure. Consequently, an optimal strategy must balance initial search efficiency with post-discovery traversal time. To minimize total completion time, a strategy might sacrifice pure search speed, for instance, by having robots temporarily pause their search to visit nearby discovered treasures, or by constraining their search paths to stay closer to known targets. Despite the geometric simplicity of our domain, this interplay between discovery and traversal yields rich, non-trivial results that establish a solid foundation for studying this problem in more complex environments.

To establish the foundational limits of our search model, we first examine an \emph{auxiliary} optimal routing problem involving a single robot starting on the circle with complete prior knowledge of all treasure locations. Although originally introduced to derive tight bounds for the primary search problem, this auxiliary variant proves to be of independent theoretical interest. We provide an exact solution for this variant, demonstrating that even in this deceptively simple setup, characterizing the worst-case placement of treasures is highly non-trivial.

\subsection{Definitions} 

Fix $k, n \in \N = \{1, 2, \ldots \}$. Suppose that $k$ \emph{treasures} are hidden on the unit circle $\mathcal{C}$ centered at the origin $(0,0)$ in the Euclidean plane. Initially, $n$ \emph{robots} are located at the origin. Each robot can move independently with a maximum speed of 1, and their movement is \emph{not} restricted to $\mathcal{C}$. A treasure is \emph{discovered} when a robot's position coincides with it. Upon discovery, all robots instantly learn the treasure's location through communication; however, each robot must individually \emph{visit} the treasure's locations to satisfy the objective. The goal is to determine a strategy that minimizes the time required for every robot to visit every treasure.

Formally, let $\mathcal{C}^k$ denote the set of all possible configurations of $k$ treasures on the circle~$\mathcal{C}$. Let $\mathcal{S}_k^n$ be the set of all valid strategies for $n$ robots. For a given strategy $S \in \mathcal{S}_k^n$ and a placement of treasures $\mathbf{T} = (T_1, T_2, \dots, T_k) \in \mathcal{C}^k$, let $\rho_S(\mathbf{T})$ be the time required for all robots to visit all treasures. The worst-case cost of strategy $S$ is then defined as:
$$
\rho_S = \sup_{\mathbf{T} \in \mathcal{C}^k} \rho_S(\mathbf{T}).
$$
Our objective is to identify optimal strategies and determine the minimax cost $\rho_k^n$, defined as:
$$
\rho_k^n = \inf_{S \in \mathcal{S}_k^n} \rho_S = \inf_{S \in \mathcal{S}_k^n} \sup_{\mathbf{T} \in \mathcal{C}^k} \rho_S(\mathbf{T}).
$$

\medskip

It is evident that the task becomes more demanding as the number of treasures increases or the number of available robots decreases. From the definition of $\rho_k^n$, we directly obtain the following monotonicity properties:
\begin{align*}
\text{ for a fixed } n \in \N, \qquad & \rho_1^n \le \rho_2^n \le \ldots \le \rho_k^n \le \ldots \\
\text{ for a fixed } k \in \N, \qquad & \rho_k^1 \ge \rho_k^2 \ge \ldots \ge \rho_k^n \ge \ldots
\end{align*}

\subsection{Auxiliary Problem}\label{sec:aux_problem_def} 

Before stating the main results, we introduce a related auxiliary problem. Consider again the unit circle $\mathcal{C}$ centered at the origin. In this variant, a single robot is placed at a point $R \in \mathcal{C}$, and $k \in \N$ treasures are located at points $T_1, T_2, \dots, T_k \in \mathcal{C}$. Without loss of generality, we may assume $R = (1,0)$. Since we are interested in the worst-case placement of treasures, we may also assume that the $k$ treasures are placed in unique locations and that no treasure is located in $R$. Let $d(A,B)$ denote the Euclidean distance between any two points $A=(a_1,a_2), B=(b_1,b_2) \in \mathbb{R}^2$:
$$
d(A,B) = \sqrt{ (a_1-b_1)^2 + (a_2-b_2)^2 }. 
$$

In this variant, the robot is immediately aware of all treasure locations and seeks to visit them in the minimum possible time. As before, the robot may move freely in the plane. The optimal cost for a given configuration of treasures is the length of a shortest path starting at $R$ and passing through all points in the set $\{T_1, \ldots, T_k\}$:
$$
C(T_1, T_2, \ldots, T_k) = \min_\sigma \left( d(R,T_{\sigma(1)}) + \sum_{i=2}^k d(T_{\sigma(i-1)}, T_{\sigma(i)}) \right),
$$
where the minimum is taken over all permutations $\sigma$ of the set $\{1, 2, \dots, k\}$ associated with $k!$ possible orders of visiting the treasures. A strategy (not necessarily unique) associated with a permutation $\sigma$ yielding $C(T_1, T_2, \ldots, T_k)$ will be called \emph{optimal}.

We are interested in characterizing the maximum cost associated with the worst-case placement of the $k$ treasures. Specifically, for a given $k \in \mathbb{N}$, we aim to compute or estimate the following constant:
\begin{equation}\label{eq:ck}
c_k = \sup_{(T_1, T_2, \ldots, T_k) \in \mathcal{C}^k} C(T_1, T_2, \ldots, T_k),
\end{equation}
where the supremum is taken over all possible configurations of $k$ points on the circle $\mathcal{C}$.

\medskip

We may restrict our attention to $k \ge 2$, as the case $k=1$ is trivial: the maximum cost $c_1 = 2$ is achieved when the single treasure is placed at $T_1 = (-1,0)$, the point antipodal to~$R$. It follows directly from the definition of $c_k$ that the sequence is non-decreasing, as visiting $k+1$ treasures is at least as demanding as visiting $k$ of them. Finally, since the robot can always visit every treasure by traversing the entire circumference of $\mathcal{C}$, the cost is bounded by $2\pi$. As a result, for all $k \in \mathbb{N}$,
$$
2 = c_1 \le c_2 \le \ldots \le c_k \le 2\pi \approx 6.2832.
$$

\subsection{Summary of Our Results} 

Let us start with presenting our results for the auxiliary problem. To state the main result we need to introduce some technical definition. It looks complicated and mysterious but everything will become clear in \cref{sec:exact_ck} so stay tuned. 

For a given integer $k \ge 2$, 
$$
s_k := \max_{m_k \le \alpha \leq M_k} \left( 4 \sum_{\ell=1}^j \sin \left( 2^{\ell-2} \alpha \right) + 2 (k-2j+1) \sin \left( \frac {\pi - (2^j-1) \alpha}{k+1-2j} \right) - 2 \sin(\alpha/2) \right),
$$
where
$$
m_k := \frac {\pi}{2^{\lfloor k/2 \rfloor-1} (k-2 \lfloor k/2 \rfloor+3)-1}, \qquad \qquad 
M_k := \frac {2\pi}{k+1},
$$
and $j = j(k, \alpha)$ is the unique integer in $\{1, 2, \ldots, \lfloor k/2 \rfloor\}$ that satisfies:
$$
\frac {2\pi}{2^j (k-2j+3)-2} \le \alpha < \frac {2\pi}{2^{j-1} (k-2(j-1)+3)-2}.
$$
(For $\alpha=M_k$, we set $j=1$ which is an edge case where we have $\alpha = \frac {2\pi}{2^{j-1} (k-2(j-1)+3)-2}$.) Note that $s_k$ is well defined as it is the maximum of a continuous function over a compact set.

\medskip

Our main result for the auxiliary problem (and, in fact, the strongest result in the entire paper) is to show that this implicit definition of $s_k$ matches $c_k$: for any integer $k \ge 2$, $c_k = s_k$. (See \cref{thm:ck=sk} in \cref{sec:exact_ck}.) Constants $s_k$ are implicitly defined but can be easily approximated numerically. The values of $s_k$ for $2 \le k \le 24$ are approximated in \cref{table:optimal} in Appendix \ref{sec:appendix_tables_sk}. 

\medskip

We independently prove some explicit upper and lower bounds: for any integer $k \ge 2$,
\begin{eqnarray*}
c_k &\ge& 2(k-1)\sin \left( \frac {\pi}{k} \right) + 2 \sin \left( \frac {\pi}{2k} \right) 
~~\ge~~  2 \pi - \frac {\pi}{k} - \frac {\pi^3}{3k^2} + \frac {7\pi^3}{24k^3} \\
c_k &\le& 2 k \sin \left( \frac {\pi}{k} \right) 
~~\le~~ 2\pi - \frac{\pi^3}{3k^2} + \frac{\pi^5}{60 k^5}. 
\end{eqnarray*}
(See \cref{lem:general_LB_ck} in \cref{sec:lb_ck} for the lower bound and \cref{lem:general_UB_ck} in \cref{sec:ub_ck} for the upper bound.) These bounds, in particular, imply that $\lim_{k \to \infty} c_k = 2\pi$.

\bigskip

Let us now present our results for the original problem. It is easy to see that $\rho_k^1 = 1+2\pi \approx 7.2832$ (see \cref{sec:n2k1}). Determining the value of $\rho_1^2$ is already non-trivial. We show that $\rho_1^2 = 1 + \frac {2\pi}{3} + \sqrt{3} \approx 4.8264$ (see \cref{thm:rho12} in \cref{sec:n2k1}). These are the only two values that we managed to determine exactly. With quite a bit of effort, we only proved the following bonds for $\rho_2^2$: $5.5675 \approx 1+\frac {\pi}{3} + c_2 \le \rho_2^2 \le 6.2195$, where $6.2195$ is an upper bound for some implicitly defined constant $U$; see~\eqref{eq:def_U}. We prove it in \cref{thm:n2k2} in \cref{sec:n2k2}.

\medskip

We independently prove some general explicit upper and lower bounds for $\rho_k^n$: for any $n, k \in \N$, 
\begin{eqnarray*}
\text{(see \cref{lem:rho_lower})} \qquad \rho_k^n &\ge& 1 + c_k \\
\text{(see \cref{thm:org_LB}, equation \eqref{eq:org_LB1})} \qquad \rho_k^n &\ge& 1 + \frac {2\pi}{(k+1)n}+2k \sin \left( \frac {\pi}{k+1} \right) \\
\text{(see \cref{lem:rho_upper}, equation \eqref{eq:rho_upper_1})} \qquad \rho_k^n &\le& 1 + 2 \pi \\
\text{(see \cref{lem:rho_upper}, equation \eqref{eq:rho_upper_2})} \qquad \rho_k^n &\le& 1 + \frac {2\pi}{n} + c_k.
\end{eqnarray*}
These bounds make a connection between the auxiliary problem and the original one. There are two lower and two upper bounds and each one provides the best bound for some pair of $n$ and $k$. The corresponding bounds for $\rho_k^n$ for $1 \le n \le 5$ and $2 \le k \le 24$ are approximated in \cref{table:bounds} in Appendix \ref{sec:appendix_tables_rho}. In particular, these bounds allow us to understand asymptotic behaviour of $\rho_k^n$ when $k \to \infty$ or $n \to \infty$: 
for any $k \in \N$, $\lim_{n \to \infty} \rho_k^n = 1 + c_k$ (see \cref{cor:n_to_infty}), and 
for any $n \in \N$, $\lim_{k \to \infty} \rho_k^n = 1 + 2\pi$ (see \cref{cor:k_to_infty}).

\subsection{Related Work} 

Search problems involving mobile agents have been extensively studied in theoretical computer science and robotics. In geometric domains, these agents are commonly modelled as autonomous robots and must coordinate their trajectories to search an environment efficiently.

When operating without a prior map of the environment, research often focuses on online exploration strategies~\cite{albers2000,albers2002,deng1991,hoffmann2001}. In many applications, however, complete exploration is not the ultimate objective, but rather an intermediate step to facilitate target discovery or task completion~\cite{kleinberg1994,papadimitriou1991}.

Conversely, when the environment's geometry is known beforehand, the focus shifts to locating hidden targets within the domain. Such targets may be stationary, as in the classic cow-path problem~\cite{beck1964,bellman1963} or parallel searching in the plane~\cite{baeza1993,baeza1995}. Alternatively, targets may be mobile, giving rise to game-theoretic formulations such as pursuit-evasion or cops-and-robbers games~\cite{bonato2017,nahin2012}.

\bigskip

To the best of our knowledge, the specific variant studied in this paper has not been previously investigated. However, it is closely related to the well-studied \emph{evacuation problem}. In an evacuation scenario, there are $n$ robots and $k$ exits; however, while our objective requires every robot to visit all $k$ targets, evacuation requires every robot to reach one of the exits.

The evacuation literature explores various communication models and spatial constraints. In~\cite{czyzowicz2014}, the authors study evacuation with unknown exit locations ($k=2$ and $k=3$) under two paradigms: \emph{wireless} communication (instantaneous long-range messaging) and \emph{face-to-face} communication (information exchange restricted to co-located robots). Subsequently, Czyzowicz et al.~\cite{czyzowicz2015} improved the bounds for two robots evacuating a disk with one exit by constructing linear and triangular detours for worst-case positions.

Other variants consider different levels of prior information, agent heterogeneity, or domain topologies. In~\cite{czyzowicz2016}, robots have a map with known exit locations on a circle's perimeter but lack their own absolute initial coordinates, knowing only their relative separation. Pattanayak et al.~\cite{pattanayak2018} investigate $k=2$ exits at unknown locations separated by a known arc distance $d$. Lamprou et al.~\cite{lamprou2016} study wireless evacuation for heterogeneous robots moving at different maximum speeds, while Borowiecki et al.~\cite{borowiecki2016} generalize the problem to graph topologies where a subset of vertices contain exits.

\section{Auxiliary Problem: Visiting $k$ Points on the Cycle} 

In this section, we completely analyze the auxiliary problem described in \cref{sec:aux_problem_def} by computing $c_k$, the cost of visiting the worst-case placement of $k$ treasures (see (\ref{eq:ck}) for a formal definition of $c_k$ and (\ref{eq:s_k}) for the definition of the associated value). In \cref{sec:connection_aux_org} we will make a connection between the auxiliary problem and the original one but here we concentrate exclusively on the auxiliary problem. 

We first provide some suboptimal but explicit lower and upper bounds for $c_k$; see \cref{sec:lb_ck} and \cref{sec:ub_ck}, respectively. To highlight the difficulty of understanding the exact values of $c_k$, we first warm up with analyzing $c_2$ in \cref{sec:c2}. The general case $k \ge 2$ is investigated in \cref{sec:exact_ck}.  The main result, \cref{thm:ck=sk}, provides an exact (but implicitly defined) value of $c_k$ for any $k \ge 2$. To show the main result we first prove some useful properties of optimal strategies in \cref{sec:optimal_strategies} before dealing independently with upper and lower bounds for $c_k$ (see \cref{sec:ck_ge_sk} and \cref{sec:ck_le_sk}, respectively). 

\subsection{Suboptimal (but Explicit) Lower Bounds} \label{sec:lb_ck}

Having noted that $c_1 = 2$, we begin with a general lower bound for $c_k$ for any integer $k \ge 2$. Although this bound is not tight (except the case $k=3$), it is explicit and straightforward to derive. Notably, it establishes that $c_k$ approaches $2\pi$ as the number of treasures grows large.

\begin{plemma}\label{lem:general_LB_ck}
For any integer $k \ge 2$, 
\begin{eqnarray}
c_k &\ge& 2(k-1)\sin \left( \frac {\pi}{k} \right) + 2 \sin \left( \frac {\pi}{2k} \right) \label{eq:ck-lower-bound} \\
&\ge& 2 \pi - \frac {\pi}{k} - \frac {\pi^3}{3k^2} + \frac {7\pi^3}{24k^3}. \label{eq:ck-lower-bound2}
\end{eqnarray}
In particular, $\lim_{k \to \infty} c_k = 2\pi$.
\end{plemma}

\begin{proof}
Fix any integer $k \ge 2$. To establish a lower bound for $c_k$, we consider a configuration where the $k$ treasures form the vertices of a regular $k$-gon inscribed in $\mathcal{C}$. Let $\alpha = 2\pi/k$ be the angular distance between consecutive treasures. We orient the $k$-gon such that the robot's starting position $R = (1,0)$ is equidistant from $T_1$ and $T_k$. Under this symmetry, the $i$-th treasure $T_i = (\cos(\theta_i), \sin(\theta_i))$ is located at the angle
$$
\theta_i = \frac {\alpha}{2} + (i-1) \alpha = \frac {\alpha (2i-1)}{2} = \frac {\pi (2i-1)}{k}.
$$
By computing the cost $C(T_1, T_2, \ldots, T_k)$ for this specific placement, we obtain a lower bound for the supremum $c_k$.

First, we observe that an optimal strategy is to visit the treasures in consecutive order. Due to the symmetry of the configuration, there are four equivalent paths: starting at either $T_1$ or $T_k$ and proceeding either clockwise or anti-clockwise around the circle. Each of these four paths has an associated cost of
\begin{eqnarray*}
d(R,T_{1}) + \sum_{i=2}^k d(T_{i-1}, T_{i}) &=& 2 \sin \left( \frac {\theta_1}{2} \right) + \sum_{i=2}^k 2 \sin \left( \frac {\theta_i - \theta_{i-1}}{2} \right) \\
&=& 2 \sin \left( \frac {\pi}{2k} \right) + 2 (k-1) \sin \left( \frac {\pi}{k} \right).
\end{eqnarray*}
To confirm that no other permutation $\sigma$ of the set $\{1, 2, \ldots, k\}$ yields a smaller cost, note that $T_1$ and $T_k$ are the treasures nearest to $R$; thus, $d(R, T_{\sigma(1)}) \ge d(R, T_1)$. Furthermore, because the treasures form a regular $k$-gon, the minimum distance between any two distinct treasures is the side length $d(T_1, T_2)$. It follows that $d(T_{\sigma(i-1)}, T_{\sigma(i)}) \ge d(T_{i-1}, T_i)$ for all $i \in \{2, 3, \ldots, k\}$. Summing these lower bounds confirms that the sequential visiting order is optimal, thereby proving inequality \eqref{eq:ck-lower-bound}.

To prove inequality~\eqref{eq:ck-lower-bound2} we need to lower bound function $\sin(x)$ for $x \in [0,\pi]$. This task is easy since $\sin(x)$ possesses a well-behaved alternating Taylor series:
$$
\sin(x) = x - \frac{x^3}{3!} + \frac {x^5}{5!} - \frac {x^7}{7!} + \ldots.
$$
Because this is an alternating series whose terms strictly decrease in absolute value when $x$ is small enough (for example, if $x \in [0,\pi]$ and the first two terms are ignored), one can truncate the series after any negative term to get a strict lower bound. In particular, truncating after the cubic term we get the following standard lower bound: for any $x \in [0,\pi]$, $\sin(x) \ge x - x^3/6$. It follows that
\begin{eqnarray*}
c_k &\ge& 2 \sin \left( \frac {\pi}{2k} \right) + 2 (k-1) \sin \left( \frac {\pi}{k} \right) \\
&\ge& 2 \left( \frac {\pi}{2k} - \frac {\pi^3}{48k^3} \right) + 2 (k-1) \left( \frac {\pi}{k} - \frac {\pi^3}{6k^3} \right) \\
&=& 2 \pi - \frac {\pi}{k} - \frac {\pi^3}{3k^2} + \frac {7\pi^3}{24k^3}.
\end{eqnarray*}
This proves inequality~\eqref{eq:ck-lower-bound2}. 

Since, trivially, $c_k \le 2\pi$ for any $k \in \N$, we get that $\lim_{k \to \infty} c_k = 2\pi$ and the proof of the lemma is finished.
\end{proof}

\subsection{Suboptimal (but Explicit) Upper Bounds} \label{sec:ub_ck}

As already mentioned a few times, $c_k \le 2\pi$ for any $k \ge 2$. Below, we provide a slightly better upper bound reducing the gap between the upper bound~\eqref{eq:ck-upper-bound} and the lower bound~\eqref{eq:ck-lower-bound} to 
$$
2 \sin \left( \frac {\pi}{k} \right) - 2 \sin \left( \frac {\pi}{2k} \right).
$$
The gap between the two corresponding weaker bounds~\eqref{eq:ck-upper-bound2} and~\eqref{eq:ck-lower-bound2} is 
$$
\left( 2\pi - \frac{\pi^3}{3k^2} + \frac{\pi^5}{60 k^5} \right) 
- \left( 2 \pi - \frac {\pi}{k} - \frac {\pi^3}{3k^2} + \frac {7\pi^3}{24k^3} \right) 
= 
\frac {\pi}{k} - \frac {7\pi^3}{24k^3} + \frac{\pi^5}{60 k^5}.
$$

\begin{plemma}\label{lem:general_UB_ck}
For any integer $k \ge 2$, 
\begin{eqnarray}
c_k &\le& 2 k \sin \left( \frac {\pi}{k} \right) \label{eq:ck-upper-bound} \\
&\le& 2\pi - \frac{\pi^3}{3k^2} + \frac{\pi^5}{60 k^5}. \label{eq:ck-upper-bound2}
\end{eqnarray}
\end{plemma}

\begin{proof}
Suppose that $k$ treasures are located at $T_i = (\cos(\theta_i), \sin(\theta_i))$, $i \in \{1, 2, \ldots, k\}$, where $0 < \theta_1 < \theta_2 < \ldots < \theta_k < 2\pi$. Fix $\alpha_1 = \theta_1$ and let $\alpha_i = \theta_i - \theta_{i-1}$, $i \in \{2, 3, \ldots, k\}$, be the angular distance between consecutive treasures. Clearly, visiting treasures following a~counter-clockwise order is not better than following the optimal order, yielding the following: 
\begin{eqnarray*}
C(T_1, T_2, \ldots, T_k) &=& \min_\sigma \left( d(R,T_{\sigma(1)}) + \sum_{i=2}^k d(T_{\sigma(i-1)}, T_{\sigma(i)}) \right) \\
&\le& d(R,T_{1}) + \sum_{i=2}^k d(T_{i-1}, T_{i}) ~~=:~~ U(T_1, T_2, \ldots, T_k).
\end{eqnarray*}
We immediately get that
$$
c_k = \sup_{(T_1, T_2, \ldots, T_k) \in \mathcal{C}^k} C(T_1, T_2, \ldots, T_k) \le \sup_{(T_1, T_2, \ldots, T_k) \in \mathcal{C}^k} U(T_1, T_2, \ldots, T_k) =: u_k.
$$

Computing $c_k$ is non-trivial but computing $u_k$ is easy. In the latter case, the optimal placement of the treasures is to place $T_k$ infinitesimally close to $R$ (on the clockwise side) and place the remaining $k-1$ points as vertices of the regular $k$-gon inscribed in $\mathcal{C}$ where the last point is $R$. 

Formally, suppose that $\theta_k = 2\pi - \epsilon$ for some arbitrarily small $\epsilon > 0$. Note that 
$$
U(T_1, T_2, \ldots, T_k) = \sum_{i=1}^k 2 \sin \left( \frac {\alpha_i}{2} \right) = 2 \sum_{i=1}^k \sin \left( \frac {\alpha_i}{2} \right) 
$$
with $\sum_{i=1}^k \alpha_i/2 = \theta_k / 2 = \pi - \epsilon/2 < \pi$. Since $\sin(x)$ is a concave function on $(0,\pi)$, it follows from Jensen's Inequality that 
$$
\frac {1}{k} \sum_{i=1}^k \sin \left( \frac {\alpha_i}{2} \right) \le \sin \left( \frac {\sum_{i=1}^k \alpha_i/2}{k} \right) = \sin \left( \frac {\pi-\epsilon/2}{k} \right).
$$
We get that 
$$
U(T_1, T_2, \ldots, T_k) \le 2 k \sin \left( \frac {\pi-\epsilon/2}{k} \right),
$$
and the upper bounds is achieved when $\alpha_1 = \alpha_2 = \ldots = \alpha_k = \frac{2\pi-\epsilon}{k}$.
In other words, given that $\theta_k = 2\pi - \epsilon$ for some arbitrarily small $\epsilon > 0$, the worst-case placement for the remaining treasures is to distribute them evenly. We conclude that 
$$
u_k = \lim_{\epsilon \to 0} 2 k \sin \left( \frac {\pi-\epsilon/2}{k} \right) = 2 k \sin \left( \frac {\pi}{k} \right),
$$
which proves~\eqref{eq:ck-upper-bound}.

\medskip

As observed in the proof of \cref{lem:general_LB_ck}, one can truncate the alternating Taylor series of $\sin(x)$ after any positive term to get a strict upper bound. In particular, for any $x \in [0,\pi]$, $\sin(x) \le x - x^3/6 + x^5/120$. Applying this inequality, we get that 
$$
c_k\leq 2\pi - \frac{\pi^3}{3k^2} + \frac{\pi^5}{60 k^5}.
$$
This proves~\eqref{eq:ck-upper-bound2} which finishes the proof of the lemma.
\end{proof}

\subsection{Warming Up with the Case $k=2$} \label{sec:c2}

\cref{lem:general_LB_ck} implies $c_2 \ge 2+\sqrt{2} \approx 3.4142$ whereas \cref{lem:general_UB_ck} implies that $c_2 \le 4$. The true value of $c_2$ is in between the two bounds, slightly larger than the above lower bound.

\begin{plemma}
The worst-case cost for $k=2$ is given by:
$$
c_2 = \frac{ \sqrt{33}+3 }{16} \sqrt{ 30 + 2\sqrt{33}} \approx 3.5203.
$$
\end{plemma}

\begin{proof}
Let $T_1, T_2 \in \mathcal{C}$ be the locations of two treasures. Without loss of generality, we may assume $d(R,T_1) \le d(R,T_2)$, meaning the robot is initially at least as close to $T_1$ as it is to $T_2$. The cost to visit both is thus:
$$
C(T_1, T_2) = \min \Big( d(R,T_1) + d(T_1,T_2), d(R,T_2) + d(T_2+T_1) \Big) = d(R,T_1) + d(T_1,T_2).
$$
By symmetry, we may set $T_1 = (\cos(\alpha), \sin(\alpha))$ for some $\alpha \in [0, \pi]$, so that $d(R, T_1) = 2\sin(\alpha/2)$.

\medskip

Let us consider two cases for a possible angle $\alpha \in [0, \pi]$. 

\medskip

\noindent \textbf{Case 1. $\alpha \le \pi/2$}:
To maximize $C(T_1, T_2)$, $T_2$ should be placed at the antipode of $T_1$, that is, $T_2 = (\cos(\alpha+\pi), \sin(\alpha+\pi))$. This yields a distance $d(T_1, T_2) = 2$ and a total cost of $2 \sin(\alpha/2) + 2$. The maximum on this interval is $2+\sqrt{2} \approx 3.4142$, occurring at the end of the range for parameter $\alpha$, at $\alpha = \pi/2$.

\medskip

\noindent \textbf{Case 2. $\alpha \ge \pi/2$}:
In this range, the cost is maximized when $T_2$ is placed at the reflection of $T_1$ across the $x$-axis, that is, $T_2 = (\cos (\alpha), - \sin (\alpha))$, yielding $d(T_1, T_2) = 2 \sin \alpha$. (Recall that it is assumed that $d(R,T_1) \le d(R,T_2)$.) The cost function becomes:
$$
f(\alpha) := 2 \sin(\alpha/2) + 2 \sin(\alpha).
$$ 
To find the extremum, we differentiate with respect to $\alpha$:
\begin{eqnarray*}
f'(\alpha) &=& \cos(\alpha/2) + 2 \cos(\alpha) \\
&=& \cos(\alpha/2) + 2 ( 2 \cos(\alpha/2)^2 - 1) \\
&=& 4 x^2 + x - 2, 
\end{eqnarray*}
where $x = \cos(\alpha / 2)$. Setting $f'(\alpha) = 0$ and applying the quadratic formula we conclude that the maximum over $\alpha \in [\pi/2, \pi]$ is obtained when $\alpha = \alpha_0 \approx 1.8719$ that satisfies $\cos(\alpha_0 / 2) = \frac {\sqrt{33}-1}{8}$. At this critical point, the associated cost is equal to
\begin{eqnarray*}
f(\alpha_0) &=& 2 \sin( \alpha_0/2) + 2 \sin(\alpha_0) \\
&=& 2 \sin( \alpha_0/2) + 4 \sin(\alpha_0/2) \cos(\alpha_0/2) \\
&=& 2 \sin( \alpha_0/2) \Big( 1 + 2 \cos(\alpha_0/2) \Big) \\
&=& 2 \sqrt{ 1 - \cos^2( \alpha_0/2)} \Big( 1 + 2 \cos(\alpha_0/2) \Big) \\
&=& \frac{ \sqrt{33}+3 }{16} \sqrt{ 30 + 2\sqrt{33}} \approx 3.5203.
\end{eqnarray*}

Comparing the two cases, we conclude that 
$$
c_2 = \max \left( \sqrt{2}+2, \frac{ \sqrt{33}+3 }{16} \sqrt{ 30 + 2\sqrt{33}} \right) = \frac{ \sqrt{33}+3 }{16} \sqrt{ 30 + 2\sqrt{33}},
$$
which finishes the proof of the lemma.
\end{proof}

\subsection{Exact (but Implicit) Values} \label{sec:exact_ck}

The case $k=2$ was relatively straightforward to investigate, yet the worst-case cost and its corresponding adversarial placement of treasures were non-trivial. This complexity suggests that finding an explicit, closed-form expression for $c_k$ for an arbitrary $k$ is unlikely. Nevertheless, numerical approximations of $c_k$ and corresponding placements of treasures reveal a distinct structural pattern, leading us to define and investigate the following family of configurations. The code used to identify the pattern can be found in Appendix \ref{sec:appendix_code_special_family}. 

\begin{pdefinition}[Special family of treasure placements]\label{def:special_family}
Fix an integer $k \ge 2$ and suppose that $\alpha \in \R$ satisfies the following bounds:
\begin{eqnarray*}
\alpha & \ge & m_k = \frac {2\pi}{2^{\lfloor k/2 \rfloor} (k-2 \lfloor k/2 \rfloor+3)-2} = \frac {\pi}{2^{\lfloor k/2 \rfloor-1} (k-2 \lfloor k/2 \rfloor+3)-1} \\
\alpha & \leq & M_k = \frac {2\pi}{2^{0} (k-2 \cdot 0+3)-2} = \frac {2\pi}{k+1}.
\end{eqnarray*}

Let $j = j(k, \alpha)$ be the unique integer in $\{1, 2, \ldots, \lfloor k/2 \rfloor\}$ that satisfies:
\begin{equation}
\frac {2\pi}{2^j (k-2j+3)-2} \le \alpha < \frac {2\pi}{2^{j-1} (k-2(j-1)+3)-2}. \label{eq:bound_on_alpha}
\end{equation}
Next, we define the sequence of angular increments $\beta_\ell$ as follows:
\begin{equation*}
\beta_\ell =
\begin{cases}
2^{\ell-1} \alpha & \text{ if } \ell \in \{1, 2, \ldots, j \} \\
\frac {2\pi - 2(2^j-1) \alpha}{k+1-2j} & \text{ if } \ell \in \{j+1, j+2, \ldots, \lfloor k/2 \rfloor \}.
\end{cases}
\end{equation*}

A placement of treasures $\textbf{T} = (T_1, T_2, \ldots, T_k) \in \mathcal{C}^k$ is called \emph{symmetric} if $T_i = (\cos(\alpha_i), \sin(\alpha_i))$ for some angles $\alpha_i \in (0, 2\pi)$ satisfying 
\begin{equation}\label{def:symmetric}
\alpha_{k+1-i} = 2\pi - \alpha_i \qquad \qquad \text{ for any } i \in \{1, 2, \ldots, \lceil k/2 \rceil \}.
\end{equation}
(Note that if $k$ is odd, then this condition implies that $\alpha_{\lceil k/2 \rceil} = \pi$.)

A symmetric placement of treasures $\textbf{T} \in \mathcal{C}^k$ is called \emph{$(k, \alpha)$-special} if the angles additionally satisfy
$$
\alpha_ i = \sum_{\ell=1}^i \beta_\ell \qquad \qquad \text{ for any } i \in \{1, 2, \ldots, \lfloor k/2 \rfloor \}.
$$
\end{pdefinition}

Although the definition of a $(k, \alpha)$-special placement may look mysterious, the underlying geometric intuition is quite straightforward. (See \cref{fig:special_placement}.) The angular distances between consecutive treasures initially double with each step. At a specific threshold (determined by $j=j(k, \alpha)$), this doubling stops, and all remaining treasures are distributed equidistantly. The stopping point $T_{j+1}$ has the following crucial property: doubling the angle stops at the very first time when the remaining equal distribution of points would produce an angle that is smaller than or equal to what doubling would give. In particular, we do not double the angle at $\ell=j+1$ since $\frac{2\pi - 2(2^{\ell-1}-1) \alpha}{k+1-2(\ell-1)} = \frac{2\pi - 2(2^{j}-1) \alpha}{k+1-2j}$ is at most $2^{\ell-1} \alpha = 2^j \alpha$ (so we fix $\beta_\ell = \frac{2\pi - 2(2^{j}-1) \alpha}{k+1-2j}$). This property holds due to the lower bound on $\alpha$ (see~\eqref{eq:bound_on_alpha}). On the other hand, we still double the angle at $\ell=j$ since $\frac{2\pi - 2(2^{\ell-1}-1) \alpha}{k+1-2(\ell-1)} = \frac{2\pi - 2(2^{j-1}-1) \alpha}{k+3-2j}$ is grater than $2^{\ell-1} \alpha = 2^{j-1} \alpha$ (so we fix $\beta_\ell = 2^{\ell-1} \alpha$). This property holds due to the upper bound on $\alpha$ (see~\eqref{eq:bound_on_alpha}). To guarantee that such threshold value is reached (that is, that $1 \le j \le \lfloor k/2 \rfloor$) we impose some bounds on $\alpha$, namely, $m_k \le \alpha \leq M_k$. Finally, note that the doubling angle implies that for $1 \leq i < j$ we have that $d(T_i, T_k) = d(T_i,T_{i+1})$ whereas the stopping rule implies that $d(T_{j}, T_k) > d(T_{j},T_{j+1})$. 

\begin{figure}[ht!]
    \centering
    \includegraphics[width=0.45\linewidth]{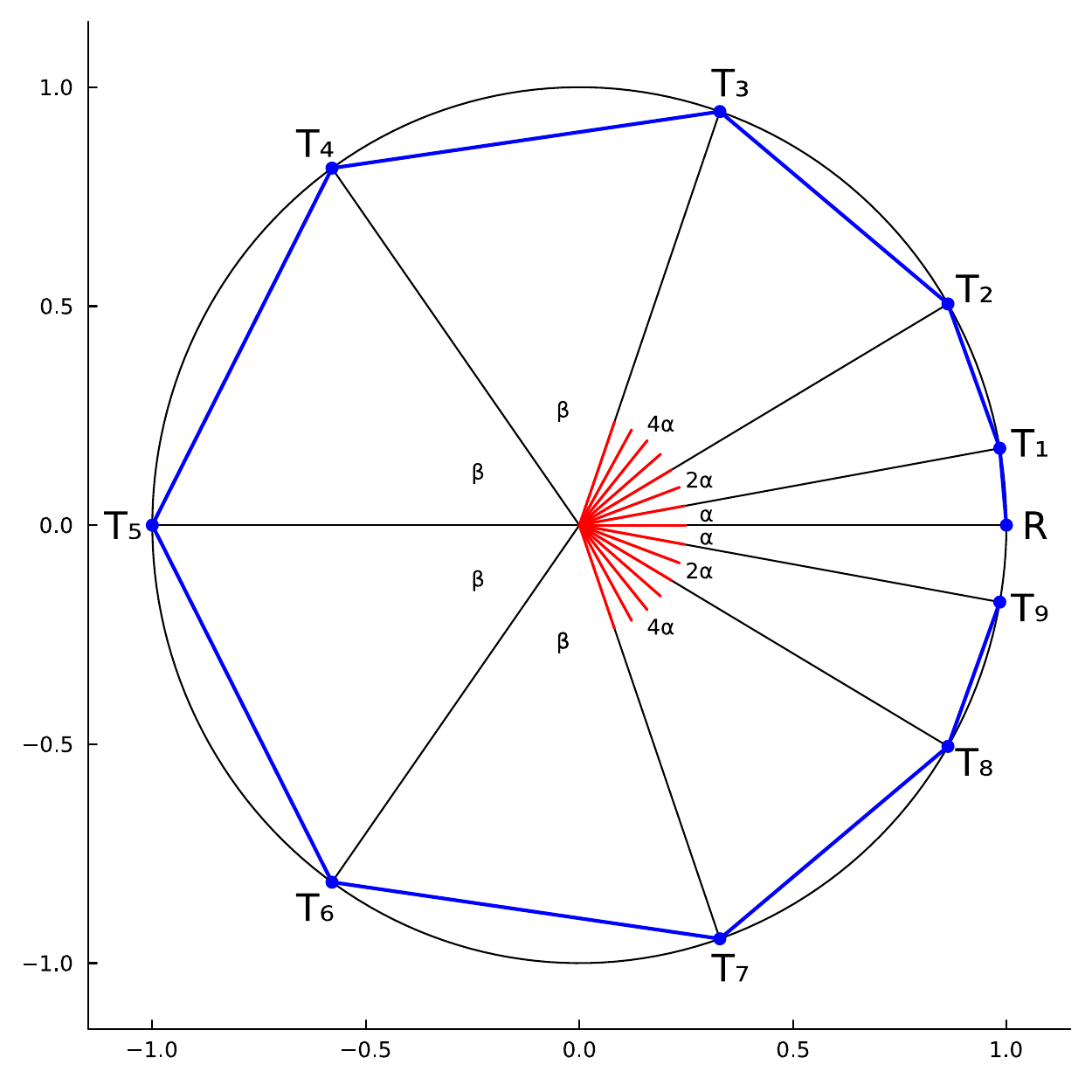}
    \includegraphics[width=0.45\linewidth]{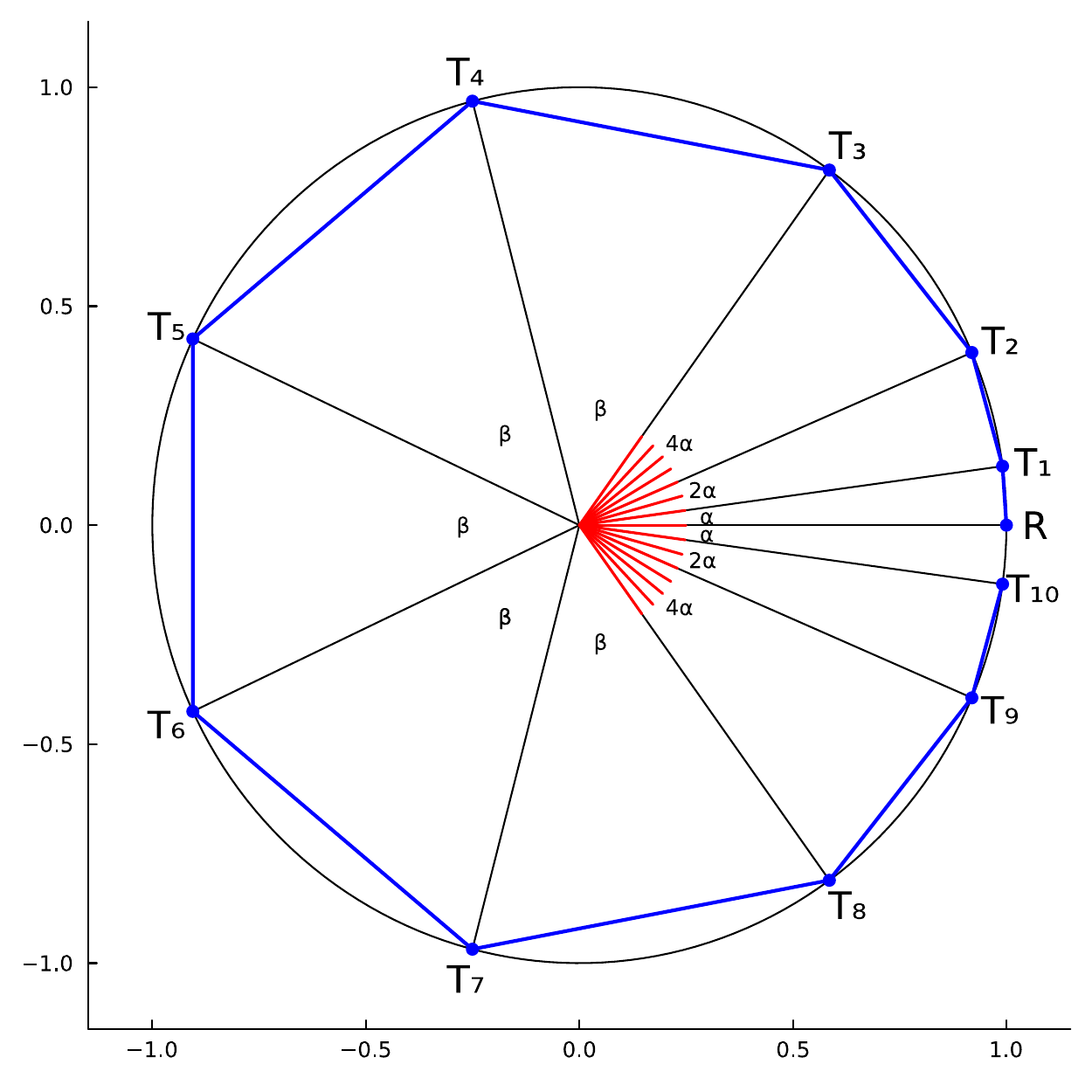}
    \caption{Examples of $(9,\alpha)$-special placement (Left) and $(10,\alpha)$-special placement (Right). In both cases, the transition occurs at $j=j(k, \alpha)=3$.}
    \label{fig:special_placement}
\end{figure}

\medskip

Consider any $(k, \alpha)$-special placement of treasures $\textbf{T} = (T_1, T_2, \ldots, T_k) \in \mathcal{C}^k$. Let $s_{k, \alpha}$ denote the cost of visiting $\textbf{T}$ in the counter-clockwise order implied by the identity permutation $\sigma(i)=i$. This cost is given by the following formula
\begin{eqnarray}
s_{k, \alpha} &=& \sum_{\ell=1}^k 2 \sin \left( \beta_\ell/2 \right) \nonumber \\
&=& 4 \sum_{\ell=1}^j \sin \left( 2^{\ell-2} \alpha \right) + 2 (k-2j+1) \sin \left( \frac {\pi - (2^j-1) \alpha}{k+1-2j} \right) - 2 \sin(\alpha/2). \label{eq:s_ka}
\end{eqnarray}
Finally, we define $s_k$ as the maximum cost over all valid initial angles $\alpha$:
\begin{equation}\label{eq:s_k}
s_k = \max_{m_k \le \alpha \leq M_k} s_{k, \alpha}.
\end{equation}
In \cref{table:optimal} (see \cref{sec:appendix_tables_sk}), approximated values of $s_k$ (for $k \in \{2, 3, \ldots, 24\}$) can be found, together with the corresponding optimal values of $\alpha$ and $j$. They were generated using the code that can be found in Appenix \ref{sec:appendix_code_sk}.

\medskip

Our ultimate goal is to establish that this carefully constructed configuration strictly captures the worst-case cost.

\begin{ptheorem}\label{thm:ck=sk}
For any integer $k \ge 2$, $c_k = s_k$.
\end{ptheorem}

After proving some useful properties in \cref{sec:optimal_strategies}, we independently prove a lower bound (see \cref{thm:ck_ge_sk} in \cref{sec:ck_ge_sk}) and an upper bound (see \cref{thm:ck_le_sk} in \cref{sec:ck_le_sk}).

\subsubsection{Properties of Optimal Strategies and Optimal Location of Treasures} \label{sec:optimal_strategies}

Before we state useful property of any optimal strategy, let us start with some basic geometric fact that we will use several times.

\begin{plemma}\label{le:segment}
Consider a circle and three distinct points $A$, $B$, and $C$ lying on this circle. Let $\stackrel{\frown}{BC}$ denote a segment of the circle with end points $B$ and $C$ that does not contain $A$. It follows that $\operatorname*{argmin}_{D\in\stackrel{\frown}{BC}}d(A,D)\subseteq\{B,C\}$. In other words, a point closest to $A$ in $\stackrel{\frown}{BC}$ is $B$ or $C$.
\end{plemma}
\begin{proof}
Consider a circle of radius $r > 0$ with a center at point $O$ and three distinct points $A$, $B$, and $C$ lying on this circle. Consider any $D\in\stackrel{\frown}{BC}$. Note that $d(A,D) = 2 r \sin(\measuredangle AOD / 2)$, where $\measuredangle AOD / 2 \in (0,\pi)$. Since the $\sin$ function is strictly concave in this domain, $d(A,D)$ is minimized at $\measuredangle AOD = \measuredangle AOB$ or at $\measuredangle AOD = \measuredangle AOC$. This implies that
$$
\operatorname*{argmin}_{D\in\stackrel{\frown}{BC}}d(A,D)\subseteq\{B,C\},
$$
which finishes the proof of the lemma.
\end{proof}

Now, we state a simple but useful observation, following~\cite{flood1956}, that no optimal path visiting treasures crosses itself.

\begin{plemma}\label{lem:no_cross_in_optimal}
Let $\textbf{T} = (T_1, T_2, \ldots, T_k) \in \mathcal{C}^k$ be any placement of $k$ treasures. A shortest path associated with an optimal strategy of visiting $\textbf{T}$ does not intersect itself.
\end{plemma}

\begin{proof}
For a contradiction, suppose that for some placement of $k$ treasures $\textbf{T}$ there exists an optimal path visiting $\textbf{T}$ that intersects itself. Edges of such path can be partitioned into 5 parts: a path from $R$ to $T_a$ (possibly degenerated, that is, it might consists of a single vertex $R$ and no edge), an edge from $T_a$ to $T_b$, a path from $T_b$ to $T_c$, an edge from $T_c$ to $T_d$, and a path from $T_d$ to $T_e$ (possibly degenerated), where segments $T_aT_b$ and $T_cT_d$ intersect at point $X$. Consider now an alternative strategy of visiting $\textbf{T}$: follow a path from $R$ to $T_a$, then go to $T_c$, follow a path from $T_c$ to $T_b$, then go to $T_d$, and finally follow a path from $T_d$ to $T_e$. The difference between the cost of an optimal path and the cost of an alternative path is equal to:
\begin{align*}
d(T_a, T_b) & + d(T_c, T_d) - d(T_a, T_c) - d(T_b, T_d) \\
& = \Big( d(T_a, X) + d(X, T_b) \Big) + \Big( d(T_c, X) + d(X, T_d) \Big) - d(T_a, T_c) - d(T_b, T_d) \\
& = \Big( d(T_a, X) + d(X, T_c) - d(T_a, T_c) \Big) + \Big( d(T_b, X)  + d(X, T_d) - d(T_b, T_d) \Big) \\
&>0
\end{align*}
by triangle inequality applied to triangles $\Delta T_a T_c X$ and $\Delta T_b T_d X$. Note that neither of the two triangles is a degenerate case of a triangle with zero area and so strict inequality holds. We get that the cost of an alternative strategy is smaller than the cost of an optimal one which gives as the desired contradiction. 
\end{proof}

The above lemma reduces the number of potential optimal strategies for traversing $k$ treasures from $k!$ to $2^{k-1}$. Indeed, suppose that $k$ treasures are located at $T_i = (\cos(\theta_i), \sin(\theta_i))$, $i \in \{1, 2, \ldots, k\}$, where $0 < \theta_1 < \theta_2 < \ldots < \theta_k < 2\pi$. It is easy to see (for example, by induction) that at any point of any optimal strategy, there are $a$ ``counter-clockwise'' treasures $T_1, T_2, \ldots, T_a$ and $b$ ``clockwise'' treasures $T_k, T_{k-1}, \ldots, T_{k+1-b}$ already visited (for some nonnegative integers $a$, $b$ such that $a+b<k$), and the robot occupies either $T_a$ or $T_{k+1-b}$. The only potentially optimal move from there is to go to $T_{a+1}$ or go to $T_{k-b}$. We may then encode any potentially optimal strategy as a binary vector $\textbf{v}=(v_1, v_2, \ldots, v_{k-1})$ of length $k-1$. We will refer to any of such vectors as a \emph{signature} of the corresponding strategy. The first coordinate indicates whether the robot goes to $T_1$ ($v_1=0$) or it goes to $T_k$ ($v_1=1$) at time $t=1$. Then, $v_t=0$ for some $t \ge 2$ indicates that the robot continues traversing treasures along the circle at time $t$: moving from $T_a$ to $T_{a+1}$ or moving from $T_{k+1-b}$ to $T_{k-b}$. On the other hand, $v_t=1$ indicates that the robot ``jumps'' to the other side of the circle at time $t$: moving from $T_a$ to $T_{k-b}$ or moving from $T_{k+1-b}$ to $T_{a+1}$. Note that at time $t=k$ there is only one treasure left to visit so both moves are equivalent (hence, the vector has length $k-1$ not $k$).

\medskip

Our next observation is that an optimal location of treasures has to be symmetric; see \cref{def:special_family}, equation~\eqref{def:symmetric}. This is a natural and expected property, but it requires justification. 

\begin{plemma}\label{lem:symmetry}
Any optimal location of the $k$ treasures is symmetric.
\end{plemma}

\begin{proof}
Consider any optimal location of $k$ treasures $\textbf{T} = (T_1, T_2, \ldots, T_k) \in \mathcal{C}^k$ such that $T_i = (\cos(\theta_i), \sin(\theta_i))$ for some $0 < \theta_1 < \theta_2 < \ldots < \theta_k < 2\pi$. Our goal is to show that $\theta_i+\theta_{k+1-i}=2\pi$ for any $i \in \{1, 2, \ldots, \lceil k/2 \rceil \}$, that is, the location of the points is symmetric about the diameter of the circle going through the starting point $R$, see definition given by equation~\eqref{def:symmetric}.

To simplify the notation, define $\theta_0=0$, $\theta_{k+1}=2\pi$, and $T_0=T_{k+1}=R = (1,0)$. For a contradiction, suppose that the condition $\theta_i+\theta_{k+1-i}=2\pi$ is not met for some $i \in \{1, 2, \ldots, \lceil k/2 \rceil \}$. Consider any $n$-element permutation $\sigma$ of visiting the $k$ treasures. We will construct another location of treasures $\textbf{T}' = (T'_1, T'_2, \ldots, T'_k) \in \mathcal{C}^k$ (the construction does not depend on $\sigma$) such that the length of the path induced by $\sigma$ for $\textbf{T}'$ is larger than the length of an optimal path for $\textbf{T}$. This will prove that $\textbf{T}$ cannot be an optimal placement, since in $\textbf{T}'$ all paths are longer than the shortest path on $\textbf{T}$, and finish the proof.

\medskip

To show this let us first define an auxiliary collection of $k$ treasures $\textbf{T}^s = (T^s_1, T^s_2, \ldots, T^s_k)$: for any $i \in [k]$, let $T^s_i = (\cos(\theta^s_i), \sin(\theta^s_i))$ where $\theta^s_i=2\pi-\theta_{k+1-i}$. Note that $\textbf{T}$ and $\textbf{T}^s$ are symmetric about the diameter of the circle going through the starting point $R$. (See \cref{fig:symmetric_paths}.) In particular, the length $\ell$ of an optimal (shortest) path for $\textbf{T}$ is the same as for $\textbf{T}^s$. Now, define $\textbf{T}'$ as follows: for any $i \in [k]$, let $T'_i = (\cos(\theta'_i), \sin(\theta'_i))$ where $\theta'_i = (\theta_i+\theta^s_i)/2 = \pi + (\theta_i-\theta_{k+1-i}) / 2$. Note that $0 < \theta'_1 < \theta'_2 < \ldots < \theta'_k < 2\pi$. (Note also that $\textbf{T}'$ is symmetric but we will not use this property; again, see \cref{fig:symmetric_paths}.) 

\begin{figure}[ht!]
    \centering
    \includegraphics[width=0.45\linewidth]{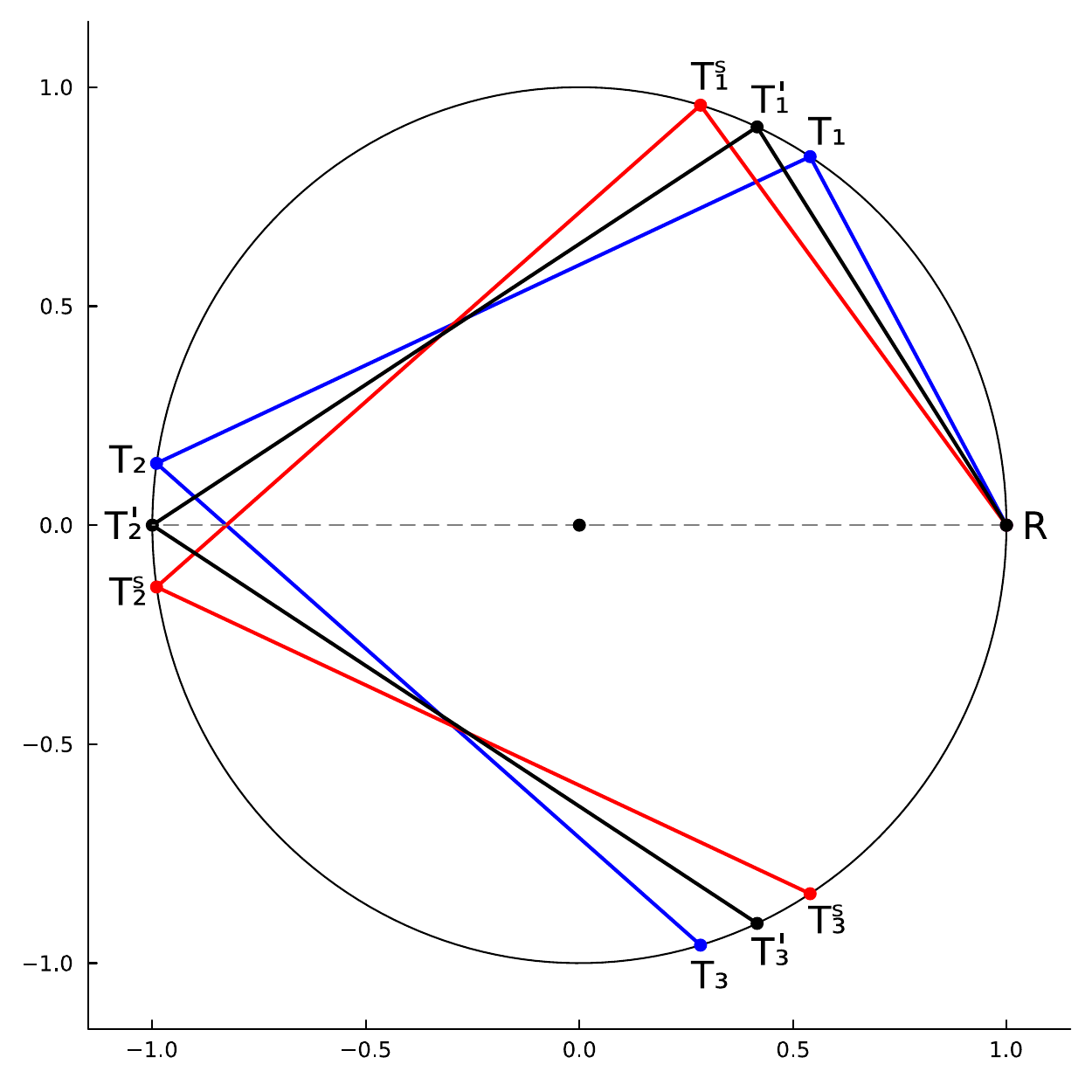}
    \caption{Example of a collection of treasures $\textbf{T}$ and its associated collections of treasures, $\textbf{T}^s$ and $\textbf{T}'$.}
    \label{fig:symmetric_paths}
\end{figure}

For any two indices $k+1 \ge i>j \ge 0$, we have the following
\begin{eqnarray*}
d(T_i,T_j) &=& 2 \sin \left( \frac{\theta_i-\theta_j}{2} \right) \\
d(T'_i,T'_j) &=& 2 \sin \left( \frac{\theta'_i-\theta'_j}{2} \right) ~~=~~ 2 \sin \left( \frac 12 \cdot \frac {\theta_i-\theta_j}{2} + \frac 12 \cdot \frac{\theta_{k+1-j}-\theta_{k+1-i}}{2} \right) \\
d(T^s_i,T^s_j) &=& 2 \sin \left( \frac {\theta_{k+1-j}-\theta_{k+1-i}}{2} \right).
\end{eqnarray*}
(Note that we included, for example, $j=0$ to cover the distance from $R=T_0$ to $T_i$.) 
Since the $\sin$ function is concave on $[0, \pi]$, in particular we have that $\sin( x/2 + y/2 ) \ge \sin(x)/2 + \sin(y)/2$ for any $x,y \in [0, \pi]$. We conclude that 
$$
d(T'_i,T'_j) ~\geq~ \frac 12 \cdot 2 \sin \left( \frac{\theta_i-\theta_j}{2} \right) + \frac 12 \cdot 2 \sin \left( \frac {\theta_{k+1-j}-\theta_{k+1-i}}{2} \right) ~=~ \frac {d(T_i,T_j) + d(T^s_j,T^s_i)}{2}.
$$
As a result, the length of the path induced by $\sigma$ for $\textbf{T}'$ is at least the average of the lengths of the two paths induced by $\sigma$ for $\textbf{T}$ and, respectively, $\textbf{T}^s$. But these two lengths are at least $\ell$ which means that the length of the path induced by $\sigma$ for $\textbf{T}'$ is at least $\ell$ as well. 

It remains to show that the length of the path induced by $\sigma$ for $\textbf{T}'$ is actually greater than $\ell$. For the contradiction, suppose that it is equal to $\ell$ which happens if and only if in every part of the path induced by $\sigma$ we have $\sin((\theta_i-\theta_j) / 2) = \sin((\theta_{k+1-j}-\theta_{k+1-i})/2)$. This implies that the angle between points $T_i$ and $T_j$ and the angle between points $T_{k+1-i}$ and $T_{k+1-j}$ are equal. However, since $\sigma$ visits all the points, the location of the $k$ treasures must be symmetric, as otherwise the required equality would not hold for the first pair of points $T_j$ and $T_{k+1-j}$ which are not symmetric (note that the starting point $T_0=T_{k+1}=R=(1,0)$ is symmetric). However, by assumption $\textbf{T}$ is not symmetric. Therefore, we get the desired contradiction which finishes the proof of the lemma.
\end{proof}

\subsubsection{Lower Bound: $c_k \ge s_k$} \label{sec:ck_ge_sk}

To establish the lower bound for $c_k$, we need to construct a valid placement of $k$ treasures $\textbf{T} \in \mathcal{C}^k$ that cannot be traversed faster than $s_k$. To this end we will prove that any $(k,\alpha)$-special placement of treasures cannot be traversed faster than by visiting the treasures in the counter-clockwise order. 

\medskip

\begin{plemma}\label{lem:lower_bound_ska}
Let $\textbf{T} = (T_1, T_2, \ldots, T_k) \in \mathcal{C}^k$ be any $(k,\alpha)$-special placement of $k \ge 2$ treasures. Then, an optimal strategy is to visit $\textbf{T}$ in the counter-clockwise order, that is,
$$
C(T_1, T_2, \ldots, T_k) = s_{k,\alpha},
$$
where $s_{k,\alpha}$ is defined in (\ref{eq:s_ka}).
\end{plemma}

The desired lower bound follows immediately from the above lemma. Indeed, the desired placement is the one that is linked with the optimal value of $\alpha$ in the definition of $s_k$ that maximizes the cost over the family of $(k,\alpha)$-special placements of treasures.

\begin{ptheorem}\label{thm:ck_ge_sk}
For any integer $k \ge 2$, $c_k \ge s_k$.
\end{ptheorem}

It remains to prove the lemma.

\begin{proof}[Proof of \cref{lem:lower_bound_ska}]
Let us consider $\textbf{T} = (T_1, T_2, \ldots, T_k) \in \mathcal{C}^k$, any $(k,\alpha)$-special placement of $k \ge 2$ treasures. As commented right after \cref{lem:no_cross_in_optimal}, we may encode any potentially optimal strategy as a binary vector $\textbf{v}=(v_1, v_2, \ldots, v_{k-1})$ of length $k-1$. In fact, by symmetry of $\textbf{T}$, without loss of generality, we may assume that $v_1=0$, that is, the robot starts by visiting $T_1$. 

Consider any vector $\textbf{v}=(0, v_2, v_3, \ldots, v_{k-1})$. Our goal is to show that the cost of a strategy associated with $\textbf{v}$ is at least $s_{k,\alpha}$, the cost of the strategy associated with vector $(0,0,\ldots,0)$ (visiting $\textbf{T}$ in the counter-clockwise order). To that end, we will show that the cost associated with vector $\textbf{v}_1=(0, v_2, v_3, \ldots, v_{i-1}, 1, 0, 0, \ldots, 0)$ is at least the one associated with vector $\textbf{v}_0=(0, v_2, v_3, \ldots, v_{i-1}, 0, 0, \ldots, 0)$. (Note that the two vectors differ only in one place, namely, $v_i$ is flipped from 1 to 0.)

\begin{figure}[ht!]
    \centering
    \includegraphics[width=0.45\linewidth]{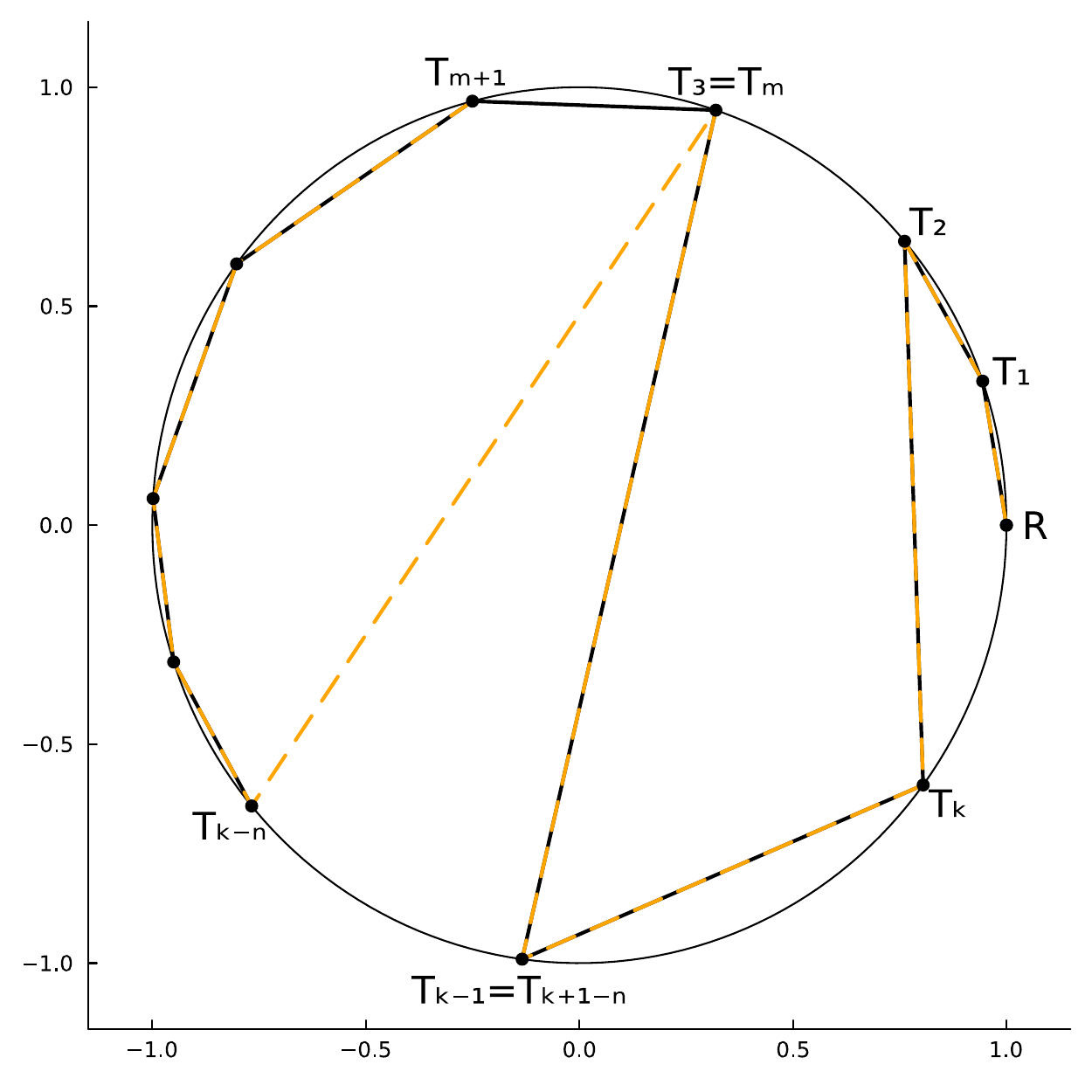}
    \caption{Two strategies to visit $k$ treasures, the first one associated with vector $\textbf{v}_0$ (black solid path) and the second one associated with vector $\textbf{v}_1$ (orange dashed line).}
    \label{fig:jumps}
\end{figure}

Consider any vector $\textbf{v}_0=(0, v_2, v_3, \ldots, v_{i-1}, 0, 0, 0, \ldots, 0)$: for some $2 \le i \le k-1$ we have $v_j \in \{0,1\}$ for all $2 \le j < i$; the remaining coordinates are equal to zero. The associated strategy visits $T_1, T_2, \ldots, T_m$ and $T_k, T_{k-1}, \ldots, T_{k+1-n}$ during the first $i-1$ rounds (for some integers $m \ge 1$ and $n \ge 0$ such that $m+n=i-1$), and then the robot continues traversing the remaining treasures along the circle. Suppose that at round $i-1$, the robot occupies $T_m$ (the proof for the other case is exactly the same). The difference between the cost associated with $\textbf{v}_1$ and $\textbf{v}_0$ is equal to
$$
d(T_m,T_{k-n}) - d(T_m, T_{m+1}).
$$
Indeed, instead of continuing from $T_m$ to $T_{m+1}$, the robot ``jumps'' to $T_{k-n}$ and continues along the circle from there; note that in what follows the robot makes exactly the same moves but in reversed order (see \cref{fig:jumps}). The difference is non-negative by the definition of the $(k,\alpha)$-special placement. To see this note that by construction $d(T_m, T_k) \geq d(T_m,T_{m+1})$ and, since the remaining point $T_{k-n}$ lies in the circle section $\stackrel{\frown\quad}{T_kT_{m+1}}$ which does not include $T_m$, we know by \cref{le:segment} that $d(T_m, T_{k-n})\geq\min\{d(T_m,T_{m+1}),d(T_m,T_k)\} = d(T_m,T_{m+1})$. Hence, any deviation from the ``continue traversing the remaining treasures along the circle'' strategy cannot decrease the cost. This finishes the proof of the lemma.
\end{proof}

\subsubsection{Upper Bound: $c_k \le s_k$} \label{sec:ck_le_sk}

To prove the upper bound $c_k \le s_k$, one needs to prove that for any valid placement of $k$ treasures $\textbf{T} \in \mathcal{C}^k$, there exists a strategy for the robot that visits all treasures in at most $s_k$ time. In fact, by \cref{lem:symmetry} we may restrict ourselves to symmetric locations of treasures; see definition given by equation~\eqref{def:symmetric}. Moreover, by \cref{lem:no_cross_in_optimal} we may restrict ourselves to paths (and the associated strategies) that do not intersect themselves. Hence, as explained right after the proof of \cref{lem:no_cross_in_optimal}, the number of potential optimal strategies reduces from $k!$ to $2^{k-2}$. Any potential optimal strategy can be represented by a \emph{signature}, a binary vector $\textbf{v}=(v_1, v_2, \ldots, v_{k-1})$ of length $k-1$. Recall that the first coordinate of $\textbf{v}$ indicates whether the robot goes to $T_1$ ($v_1=0$) or it goes to $T_k$ ($v_1=1$) at time $t=1$. At this point, since we restrict ourselves to symmetric locations of treasures, we may fix $v_1 = 0$. Then, $v_t=0$ for some $t \ge 2$ indicates that the robot continues traversing treasures along the circle at time $t$; $v_t=1$ indicates that the robot ``jumps'' to the other side of the circle at time $t$.

It is not uncommon (but always somewhat surprising) that it seems easier to prove a slightly stronger statement. Indeed, we will prove that the desired upper bound holds even if we restrict ourselves to $k-1$ simple strategies (out of $2^{k-2}$), namely, the one with no ``jumps'' (with signature $\textbf{v}=(0,0, \ldots, 0)$) or those with exactly one ``jump''. Signatures of such strategies have at most one ``$1$''. 

\medskip

Formally, consider any symmetric location of $k$ treasures $\textbf{T}=(T_1, T_2, \ldots, T_k) \in \mathcal{C}^k$: $T_i = (\cos(\theta_i), \sin(\theta_i))$, $i \in \{1, 2, \ldots, k\}$, where $0 < \theta_1 < \theta_2 < \ldots < \theta_k < 2\pi$ and for each $i \in [k]$ we have $\theta_i+\theta_{k+1-i}=2\pi$. For convenience, we set $\theta_0=0$ and $\theta_{k+1}=2\pi$, and for each $i \in [k+1]$ denote $\vartheta_i=\theta_i-\theta_{i-1}$. Note that for each $i \in [k+1]$, $\vartheta_i>0$ and $\vartheta_i=\vartheta_{k+2-i}$.

We restrict ourselves to $k-1$ strategies associated with $k-1$ paths $\mathcal{L}_\ell$ ($\ell \in \{0,1,\ldots, k-2\}$) defined as follows:
\begin{itemize}
    \item $\mathcal{L}_0=(R,T_1,\dots,T_k)$ (no ``jump'');
    \item $\mathcal{L}_\ell=(R,T_1, \dots, T_\ell, T_k, T_{k-1}, \dots, T_{\ell+1})$ for some $\ell \in [k-2]$ (``jump'' from $T_\ell$ to $T_k$ and then traverse the remaining treasures in the reversed order).
\end{itemize}
Let $d(\mathcal{L}_\ell)$ denote the length of path $\mathcal{L}_\ell$. Note that for any $\ell \in[k-2]$ we have
\begin{eqnarray}
d(\mathcal{L}_\ell) - d(\mathcal{L}_0) &=& d(T_\ell,T_k) - d(T_\ell, T_{\ell+1}) \label{eq:diff_Ll-L0} \\
&=& 2\left(\sin\left(\sum_{j=\ell+1}^k \vartheta_j/2\right)-\sin(\vartheta_{\ell+1}/2)\right) \nonumber \\
&=& 2\left(\sin\left(\vartheta_{k+1}/2+\sum_{j=1}^\ell\vartheta_j/2\right)-\sin(\vartheta_{\ell+1}/2)\right). \nonumber
\end{eqnarray}
More importantly, since the location of treasures is symmetric, if $\ell > k/2$, then $\theta_\ell \ge \pi$ and so the angular distance between $T_\ell$ and $T_k$ is less than $\pi$ and is grater than the angular distance between $T_\ell$ and $T_{\ell+1}$. As a result, we get that $d(T_\ell,T_k) > d(T_\ell, T_{\ell+1})$ and so $d(\mathcal{L}_\ell) > d(\mathcal{L}_0)$. It follows that the optimal cost $B(T_1, T_2, \ldots, T_k)$ of visiting the $k$ treasures using one of these $k-1$ specific strategies satisfies the following property: 
\begin{equation}\label{eq:b}
B(T_1, T_2, \ldots, T_k) = \min_{0 \le \ell \le k-2} d(\mathcal{L}_\ell) = \min_{0 \le \ell \le k/2} d(\mathcal{L}_\ell).
\end{equation}

Finally, let us define the maximum cost associated with the worst-case symmetric placement of the $k$ treasures for this more restrictive variant: 
\begin{equation}\label{eq:bk}
b_k = \sup_{(T_1, T_2, \ldots, T_k) \in \mathcal{C}^k} B(T_1, T_2, \ldots, T_k),
\end{equation}
where the supremum is taken over all symmetric locations of $k$ treasures on the circle $\mathcal{C}$. Since for any symmetric location of treasures $\textbf{T}$, $C(T_1, T_2, \ldots, T_k) \le B(T_1, T_2, \ldots, T_k)$, we get that $c_k \le b_k$. The desired upper bound ($c_k \le s_k$) is implied by the following stronger statement.

\begin{ptheorem}\label{thm:ck_le_sk}
For any integer $k \ge 2$, $c_k \le b_k \le s_k$.
\end{ptheorem}

Note that the $\sin$ function is strictly concave on the interval $[0, \pi]$. Hence, for any $\ell \in \{0, 1, \ldots, \lfloor k/2\rfloor\}$, $d(\mathcal{L}_\ell)$ is a strictly concave function of $(\vartheta_i/2)_{i \in [k+1]}$ as it is a sum of strictly concave functions. Since the minimum of two strictly concave functions is itself strictly concave, our objective function $B(T_1, T_2, \ldots, T_k)$ is also strictly concave. Finally, note that the domain of angles, $(\vartheta_i)_{i \in [k+1]}$, over which we optimize is a convex set, as $\vartheta_i\in[0,\pi]$, $\vartheta_i=\vartheta_{k+2-i}$, and $\sum_{i=1}^{k+1}\vartheta_i=2\pi$. This implies that the supremum in the definition of $b_k$ (see~\eqref{eq:bk}) is achieved for the unique location of $k$ treasures. (This also shows that the problem can be efficiently handled by numerical optimization algorithms.) The main ingredient of the argument is the following observation about this optimal location. 

\begin{plemma}\label{lem:no_jumps}
Consider the more restrictive variant of the problem involving only $\lfloor k/2\rfloor+1$ strategies. An optimal location of $k$ treasures $\mathbf{T}=(T_1, T_2, \ldots, T_k)$ satisfies the following property: 
$$
B(T_1, T_2, \ldots, T_k) = d(\mathcal{L}_0).
$$
\end{plemma}

Note that this lemma does not say anything about optimality of paths with ``jumps''. It only implies that one of the optimal paths must be without ``jumps'', but there could be other paths with ``jumps'' of the same length. In fact, we will see that in the optimal location of treasures this is indeed the case that there are multiple paths of the same optimal length.

\begin{proof}[Proof of \cref{lem:no_jumps}]
Consider $\mathbf{T}=(T_1, T_2, \ldots, T_k)$, an optimal and symmetric location of $k$ treasures  that yields $b_k$ (that is, $b_k = B(T_1, T_2, \ldots, T_k)$). For a contradiction, suppose that $d(\mathcal{L}_0)$ is not optimal, that is, $d(\mathcal{L}_0) > d(\mathcal{L}_s)$ for some $s \in [\lfloor k/2\rfloor]$. There could be many optimal paths so let us define $\{\mathcal{L}_s : s \in S \subseteq [\lfloor k/2\rfloor] \}$ to be the set of optimal paths. All paths in this set are of the same (optimal) length and strictly shorter than other paths, in particular, strictly shorter than $\mathcal{L}_0$. Let $\ell=\max\{S\}$ be the largest index of an optimal path. Recall that $\ell \le k/2$.

For a given $\epsilon > 0$, we modify slightly $\mathbf{T}$ to get another location of treasures $\mathbf{T}' = \mathbf{T}'(\epsilon) = (T'_1, T'_2, \ldots, T'_k)$: increase angles $\vartheta_{1}$ and $\vartheta_{k+1}$ by $\epsilon$ and decrease angles $\vartheta_{\ell+1}$ and $\vartheta_{k+1-\ell}$ by~$\epsilon$. Note that it is possible that $\ell+1=k+1-\ell$ (this \emph{special} case happens when $k$ is even and $\ell = k/2$). In such cases, we decrease $\vartheta_{\ell+1} = \vartheta_{k+1-\ell}$ by $2 \epsilon$. In other words, we get $\mathbf{T}'$ from $\mathbf{T}$ by rotating $T_{1}, T_{2}, \ldots, T_{\ell}$ slightly counter-clockwise and rotating $T_{k}, T_{k-1}, \ldots, T_{k+1-\ell}$ slightly clockwise. Note that this operation ensures that $\mathbf{T}'$ is symmetric. We will show that for some sufficiently small $\epsilon > 0$ we have that $B(T'_1, T'_2, \ldots, T'_k) > B(T_1, T_2, \ldots, T_k)$ which will contradict an optimality of $\mathbf{T}$. 

Note that the above transformation affects the length of each path in a continuous way. Hence, if $\epsilon$ is small enough, then all non-optimal paths are still longer than any optimal one. Then, it is enough to show that each path $\mathcal{L}_s$ with $s \in S$ increases its length. We will investigate the derivative of the change of the length of any optimal path. The special case ($\ell+1=k+1-\ell$) has to be treated slightly differently and, additionally, there are two cases to consider: $s=\ell$ and $s<\ell$ (depending on which case is considered, the segment from $T_\ell$ to $T_{\ell+1}$ that gets shorter is traversed by the path or not). Therefore, we have in total four cases to consider:
\begin{itemize}
    \item[(C1)] $\ell+1< k+1-\ell$ and $s<\ell$;
    \item[(C2)] $\ell+1< k+1-\ell$ and $s=\ell$;
    \item[(C3)] $\ell+1= k+1-\ell$ and $s<\ell$;
    \item[(C4)] $\ell+1= k+1-\ell$ and $s=\ell$.
\end{itemize}
For each of these cases, we compute the derivative of the path length under the infinitesimal change of the location of treasures, and argue that it is positive.

\medskip

\textbf{Case (C1)}: $\ell+1< k+1-\ell$ and $s<\ell$. In this case, the change affects the following segments of path $\mathcal{L}_s$: we 
increase the distance from $R$ to $T_1$ by angle $\epsilon$, 
increase the distance from $T_s$ to $T_k$ by angle $2\epsilon$ (this part is non-trivial but will be justified soon), 
decrease the distance from $T_{k+1-\ell}$ to $T_{k-\ell}$ by angle $\epsilon$, and
decrease the distance from $T_{\ell+1}$ to $T_{\ell}$ by angle~$\epsilon$. 
Therefore, the derivative of the path length in this case is equal to
\begin{align*}
\cos(\vartheta_1/2)&+2\cos\left( \vartheta_{k+1}/2+\sum_{j=1}^{s}\vartheta_j/2 \right)-\cos(\vartheta_{k+1-\ell}/2) -\cos(\vartheta_{\ell+1}/2) \\
& = \cos(\vartheta_1/2)+2\cos\left( \vartheta_{k+1}/2+\sum_{j=1}^{s}\vartheta_j/2 \right)-2\cos(\vartheta_{\ell+1}/2),
\end{align*}
since $\vartheta_{\ell+1} = \vartheta_{k+1-\ell}$.

Now, a crucial observation is that $\vartheta_{k+1}+\sum_{j=1}^{s}\vartheta_j < \vartheta_{\ell+1}$. Indeed, since $d(\mathcal{L}_\ell) < d(\mathcal{L}_0)$, by \eqref{eq:diff_Ll-L0} we get that $d(T_\ell,T_k) < d(T_\ell, T_{\ell+1})$. Hence, the above inequality holds for $s=\ell$ and so it clearly holds for all $s \in S$ since the left hand side gets smaller when $s < \ell$. Note that, since $\vartheta_{k+1}+\sum_{j=1}^{s}\vartheta_j < \vartheta_{\ell+1}$ and $\vartheta_{k+1}+\sum_{j=1}^{s}\vartheta_j + \vartheta_{\ell+1} < 2\pi$, we have that $\vartheta_{k+1}+\sum_{j=1}^{s}\vartheta_j<\pi$ so, indeed, the distance from $T_s$ to $T_k$ increases. Moreover, this implies that $\cos(  \vartheta_{k+1}/2+\sum_{j=1}^{s}\vartheta_j/2 )-\cos(\vartheta_{\ell+1}/2)>0$ as the $\cos$ function is a decreasing function for angles less than $\pi$. Since, trivially, $\cos(\vartheta_1/2)>0$, we conclude that the proposed change of angles increases the length of $\mathcal{L}_s$ while keeping the symmetry.

\medskip

\textbf{Case (C2)}: $\ell+1< k+1-\ell$ and $s=\ell$. This time, we 
increase the distance from $R$ to $T_1$ by angle $\epsilon$, 
increase the distance from $T_\ell$ to $T_k$ by angle $2\epsilon$, and
decrease the distance from $T_{k+1-\ell}$ to $T_{k-\ell}$ by angle $\epsilon$. 
(Note that the path from $T_{\ell+1}$ to $T_{\ell}$ that gets shorter is not traversed.) The derivative is equal to
$$
\cos(\vartheta_1/2)+2\cos\left( \vartheta_{k+1}/2+\sum_{j=1}^{s}\vartheta_j/2 \right)-\cos(\vartheta_{k+1-\ell}/2),
$$
which is strictly larger than what we had in Case (C1), and so also positive. 

\medskip

\textbf{Case (C3)}: $\ell+1= k+1-\ell$ and $s<\ell$.Note that in this case it is possible that $\vartheta_{\ell+1} \geq \pi$. However, if this happens, then all distances that change would be increased, and so the desired property would hold. Hence, we need to consider the case when $\vartheta_{\ell+1}<\pi$. This time, we 
increase the distance from $R$ to $T_1$ by angle $\epsilon$, 
increase the distance from $T_s$ to $T_k$ by angle $2\epsilon$, 
decrease the distance from $T_{k+1-\ell}=T_{\ell+1}$ to $T_{k-\ell}=T_{\ell}$ by angle $2\epsilon$ (recall that this is a special case). 
The derivative is equal to
$$
\cos(\vartheta_1/2)+2\cos\left( \vartheta_{k+1}/2+\sum_{j=1}^{s}\vartheta_j/2 \right)-2\cos(\vartheta_{\ell+1}/2),
$$
and we are back to Case~(C1).

\medskip

\textbf{Case (C4)}: $\ell+1= k+1-\ell$ and $s=\ell$. This time, we
increase the distance from $R$ to $T_1$ by angle $\epsilon$ and
increase the distance from $T_\ell$ to $T_k$ by angle $2\epsilon$.  
(Note that the path from $T_{\ell+1}$ to $T_{\ell}$ that gets shorter is not traversed.) The conclusion follows immediately, as some segments get longer but no segment gets shorter. This finishes the proof of Case~(C4) and so the proof of the lemma is finished too.
\end{proof}

With \cref{lem:no_jumps} at hand, we can come back and finalize the proof of \cref{thm:ck_le_sk}.

\begin{proof}[Proof of \cref{thm:ck_le_sk}]
Let $\mathbf{T}=(T_1, T_2, \ldots, T_k)$ be the optimal and symmetric location of $k$ treasures that yields $b_k$ (that is, $b_k = B(T_1, T_2, \ldots, T_k)$). It follows from \cref{lem:no_jumps} that $B(T_1, T_2, \ldots, T_k) = d(\mathcal{L}_0)$, that is, one optimal way to visit the $k$ treasures in $\mathbf{T}$ is to use the path with no ``jump''. There could be more optimal paths but for all $\ell \in [ \lfloor k/2 \rfloor ]$, $d(\mathcal{L}_\ell) \ge d(\mathcal{L}_0)$ or, equivalently (see \eqref{eq:diff_Ll-L0}), that 
\begin{equation}\label{eq:constraint}
\text{ for all } \ell \in [ \lfloor k/2 \rfloor ], \qquad d(T_\ell, T_{\ell+1}) \le d(T_\ell, T_k).
\end{equation}
Hence, instead of searching for the worst case location of $k$ treasures, we can redefine $b_k$ as follows (see \eqref{eq:bk}):
$$
b_k = \sup_{(T_1, T_2, \ldots, T_k) \in \mathcal{C}^k} d(\mathcal{L}_0) = \sup_{(T_1, T_2, \ldots, T_k) \in \mathcal{C}^k}  \sum_{i=1}^k 2 \sin( \vartheta_i/2),
$$
where the supremum is taken over all symmetric locations of $k$ treasures on the circle $\mathcal{C}$ that satisfy \eqref{eq:constraint}. In fact, it will be convenient to partition all possible locations into families with the same location of $T_1$: 
\begin{equation}\label{eq:new_bk}
b_k = \sup_\alpha \sup_{(T_1, T_2, \ldots, T_k) \in \mathcal{C}^k} d(\mathcal{L}_0) = \sup_\alpha \sup_{(T_1, T_2, \ldots, T_k) \in \mathcal{C}^k}  \sum_{i=1}^k 2 \sin( \vartheta_i/2),
\end{equation}
where the second suprema are taken over all symmetric locations of $k$ treasures on the circle $\mathcal{C}$ with $\vartheta_1 = \alpha$ that satisfy \eqref{eq:constraint}.

Consider any symmetric $\textbf{T}$ with $\vartheta_1 = \vartheta_{k+1} = \alpha$, that is, positions of $T_1$ and $T_k$ are fixed. If there were no additional constraints, in order to maximize $d(\mathcal{L}_0)$ one should distribute the remaining treasures evenly, that is, fix $\vartheta_2 = \vartheta_3 = \ldots = \vartheta_k = 2(\pi-\alpha)/(k-1)$. If \eqref{eq:constraint} is satisfied for $\ell = 1$ for such location of treasures, then we simply do that. However, if \eqref{eq:constraint} is not satisfied for $\ell = 1$, then it is clearly the best to place $T_2$ as far away from $T_1$ as possible but make sure that $d(T_1, T_2) \le d(T_1, T_k)$. This is achieved when $d(T_1, T_2) = d(T_1, T_k)$, that is, when $\vartheta_2$, the angular distance between $T_1$ and $T_2$ is equal to $\vartheta_1 + \vartheta_{k+1} = 2 \vartheta_1$, the sum of the angular distance between $T_1$ and $R$, and the angular distance between $R$ and $T_k$. 

We may continue this argument and conclude that the best location of treasures (in the family of symmetric locations of $k$ treasures with $\vartheta_1 = \alpha$ that satisfy \eqref{eq:constraint}) is obtained as follows. Start with $\vartheta_1 = \alpha$. If distributing evenly the remaining $k-2i$ treasures does not violate \eqref{eq:constraint}, then do it; otherwise, double the angle (fix $\vartheta_{i+1} = 2 \vartheta_i$) and continue from there. But this is exactly the definition of $(k, \alpha)$-special placement (see \cref{def:special_family}). It follows from \cref{lem:lower_bound_ska} that $\sup_{(T_1, T_2, \ldots, T_k) \in \mathcal{C}^k}  \sum_{i=1}^k 2 \sin( \vartheta_i/2)$ in \eqref{eq:new_bk} is equal to $s_{k,\alpha}$ where $s_{k,\alpha}$ defined in~(\ref{eq:s_ka}). 

There is one small caveat that needs to be discussed. In the definition of $(k, \alpha)$-special placement, it is assumed that $m_k \le \alpha \leq M_k$. This convenient bound for $\alpha$ was introduced earlier as the values of $\alpha$ outside of this interval do not yield worse-case placements of treasures. There was no justification needed for the lower bound of $c_k$ in \cref{sec:ck_ge_sk}: any placement of treasures yields a lower bound and so it makes sense to optimize parameter $\alpha$ only over a promising range for it. For the upper bound of $c_k$, we need to justify this choice. First, note that if $\alpha>M_k$, then the position of $T_1$ and $T_k$ are determined by the choice of $\alpha$ but the remaining treasures are evenly distributed. However, in this case, worse location of treasures is obtained for $\alpha=M_k$, so we may safely assume that $\alpha\leq M_k$. On the other hand, if $\alpha<m_k$, then in our argument we never reach the situation that distributing evenly the remaining treasures does not violate \eqref{eq:constraint}. Hence, no symmetric configuration of $k$ treasures $\textbf{T}$ with $\vartheta_1 = \alpha$ satisfies \eqref{eq:constraint} and so it cannot be optimal. This means that after optimizing with respect to $m_k\leq \alpha \leq M_k$ we get that $b_k = s_k$; see \eqref{eq:s_k}. This finishes the proof of the theorem. 
\end{proof}

\section{Back to the Original Problem} \label{sec:org_problem}

\subsection{Implications of the Auxiliary Problem for the Original One: Asymptotic Behaviour ($k \to \infty$ or $n \to \infty$)} \label{sec:connection_aux_org}

While the study of $c_k$ is of independent interest, its primary utility lies in providing fundamental bounds for the original multi-robot coordination problem $\rho_k^n$. In particular, it will allow us to understand asymptotic behaviour of $\rho_k^n$ when $k \to \infty$ or $n \to \infty$. We begin with the following simple lower bound.

\begin{plemma}\label{lem:rho_lower}
For any $n, k \in \N$, 
$$
\rho_k^n \ge 1 + c_k.
$$ 
\end{plemma}

\begin{proof}
Let $S \in \mathcal{S}_k^n$ be an arbitrary strategy for $n$ robots. Let $t_0$ be the first time at which any robot reaches the circle $\mathcal{C}$. Since all robots start at the origin and have a maximum speed of $1$, it is clear that $t_0 \ge 1$.

Suppose a robot $r^*$ reaches point $R \in \mathcal{C}$ at time $t_0$. At this instant, let the $k$ treasures be placed on $\mathcal{C}$ in a configuration that realizes the worst-case cost $c_k$ for a robot starting at $R$. Even if we provide all robots with full information regarding the treasure locations at time $t_0$, robot $r^*$ still requires at least $c_k$ additional units of time to visit every treasure. Because the objective requires all robots to visit all treasures, the total time $\rho_S$ cannot be less than $t_0 + c_k$. Thus, we have:
$$
\rho_k^n = \inf_S \rho_S \ge 1 + c_k,
$$
which concludes the proof.
\end{proof}

We now establish upper bounds for the optimal time $\rho_k^n$ by considering two specific search strategies.

\begin{plemma}\label{lem:rho_upper}
For any $n, k \in \N$, 
\begin{eqnarray}
\rho_k^n &\le& 1 + 2 \pi \label{eq:rho_upper_1} \\
\rho_k^n &\le& 1 + \frac {2\pi}{n} + c_k. \label{eq:rho_upper_2}
\end{eqnarray}
\end{plemma}

\begin{proof}
To prove \eqref{eq:rho_upper_1}, consider a baseline strategy where all robots move at maximum speed directly to the circle $\mathcal{C}$ and then traverse its entire circumference. Regardless of the treasure locations, every point on the circle (and thus every treasure) will be visited by all robots within $1 + 2\pi$ units of time.

\medskip

To prove \eqref{eq:rho_upper_2}, we employ a coordinated search-and-visit strategy.

\medskip \noindent \textbf{Deployment Phase}: 
Let $P_1, P_2, \ldots, P_n$ be $n$ points forming a regular $n$-gon inscribed in~$\mathcal{C}$. Each robot $i \in [n]$ moves directly to its assigned point $P_i$, arriving at time $t=1$.

\medskip \noindent \textbf{Search Phase}: 
From $P_i$, each robot $i$ moves clockwise along the arc of the circle. Since the robots are spaced at intervals of $2\pi/n$, every point on $\mathcal{C}$ will have been encountered by at least one robot after an additional $2\pi/n$ units of time.

\medskip \noindent \textbf{Completion Phase}: 
By time $1 + 2\pi/n$, the locations of all $k$ treasures have been discovered and shared among all robots. Each robot then proceeds from its current position on $\mathcal{C}$ to visit any treasures it has not yet reached. By the definition of $c_k$, this final traversal requires at most $c_k$ additional units of time.

\medskip \noindent
Summing the time for these three phases yields a total time of $1 + \frac{2\pi}{n} + c_k$, which completes the proof of the lemma.
\end{proof}

By combining the results from \cref{lem:rho_lower} and \cref{lem:rho_upper}, equation~\eqref{eq:rho_upper_2}, we can characterize the limiting behaviour of the optimal time as the number of robots $n$ grows large. This result confirms that with a sufficiently large swarm of robots, the search time $2\pi/n$ vanishes, leaving only the ``physical'' limit, that is, the time required to reach the circle plus the worst-case time for a single robot to visit all discovered treasures.

\begin{pcorollary}\label{cor:n_to_infty}
For any $k \in \N$, 
$$
\lim_{n \to \infty} \rho_k^n = 1 + c_k.
$$
\end{pcorollary}

Conversely, we can deduce the behaviour as the number of treasures $k$ tends to infinity by applying \cref{lem:rho_lower}, \cref{lem:rho_upper}, equation \eqref{eq:rho_upper_1}, and the fact that $\lim_{k \to \infty} c_k = 2 \pi$ (see \cref{lem:general_LB_ck}).

\begin{pcorollary}\label{cor:k_to_infty}
For any $n \in \N$, 
$$
\lim_{k \to \infty} \rho_k^n = 1 + 2\pi.
$$
\end{pcorollary}

\subsection{Two Robots and One Treasure ($n=2$, $k=1$)} \label{sec:n2k1}

The problem involving one robot and any number of treasures $k \in \N$ is very easy. \cref{lem:rho_upper} implies that $\rho_k^1 \le 1+2\pi$. To show a matching lower bound for $\rho_k^1$, let us fix any $\epsilon > 0$, arbitrarily small. Regardless of the searching strategy used, at time $1+2\pi-\epsilon$ there are some points in $\mathcal{C}$ that are still \emph{not} visited by the robot. One can hide some treasures there which means the job is not done yet. By taking $\epsilon \to 0$, we get that for any $k \in \N$,
$$
\rho_k^1 = 1+2\pi.
$$

\medskip

Determining the value of $\rho_1^2$ is already non-trivial. Let us start with the following observation. 

\begin{plemma}\label{lem:sqrt3jump}
Consider the unit circle $\mathcal{C}$ and two points from $\mathcal{C}$: $R_1 = (0, 1)$ (angular location $\pi/2$), $R_2 = (0, -1)$ (angular location $3\pi/2$). Let $\stackrel{\frown}{R_1R_2}$ denote a segment of the circle with end points $R_1$ and $R_2$ (half-circle) that contains point $(1,0)$ (angular location $0$). 
Suppose that two robots are at points $R_1$ and, respectively, $R_2$, and one treasure is hidden somewhere in $\stackrel{\frown}{R_1R_2}$. There exists a strategy that guarantees that both robots visit the treasure in time $\pi/6 + \sqrt{3}$. 
\end{plemma}

\begin{proof}
Suppose that both robots move toward point $(1,0)$ along the circle $\mathcal{C}$. When one of them discovers the treasure, it immediately informs the other robot which then stops traversing the circle and runs straight to the treasure.

By symmetry, we may assume that the treasure is placed at a point with angular location $\alpha$, that is, at point $T = (\cos(\alpha), \sin(\alpha))$ for some $\alpha \in [0, \pi/2]$. The above strategy finishes at time
$$
f(\alpha) = \frac {\pi}{2} - \alpha + 2 \sin( \alpha ).
$$
Since $f'(\alpha) = -1 + 2 \cos( \alpha )$, the maximum of the function $f(\alpha)$ is obtained at $\alpha = \pi/3$ and so we are guaranteed to visit the treasure before time $\pi/2 - \pi/3 + 2 \cdot (\sqrt{3}/2) = \pi/6 + \sqrt{3}$. This finishes the proof of the lemma. 
\end{proof}

With \cref{lem:sqrt3jump} at hand, we can prove the exact value of  $\rho_1^2$.

\begin{ptheorem}\label{thm:rho12}
$\rho_1^2 = 1 + \frac {2\pi}{3} + \sqrt{3} \approx 4.8264$.
\end{ptheorem}

\begin{proof}
To prove an upper bound, let us consider the following strategy. Both robots go straight to point $(-1,0)$ and then they split, one of them goes along the circle $\mathcal{C}$ to point $R_1 = (0, 1)$ (clockwise) whereas the other one goes along the circle $\mathcal{C}$ to point $R_2 = (0, -1)$ (counter-clockwise). They reach their destination at time $1 + \pi/2$. 

If the treasure is discovered at this point, then the robots inform each other about its location and the robot that is missing the treasure goes straight to it. The process ends in additional 2 units of time for a total of $3 + \pi/2 \approx 4.5708$. On the other hand, if the treasure is still not discovered, then the robots can follow a strategy guaranteed by \cref{lem:sqrt3jump} and finish their job in additional $\pi/6 + \sqrt{3}$ units of time for a total of $1+\pi/2+\pi/6 + \sqrt{3} = 1+2\pi/3+\sqrt{3} \approx 4.8264$. Hence, $\rho_1^2 \le 1 + \frac {2\pi}{3} + \sqrt{3}$.

\medskip

Let us now move to the lower bound. Fix any $\epsilon > 0$, arbitrarily small. Let $S \in \mathcal{S}_1^2$ be an arbitrary strategy for two robots. Trivially, regardless of the strategy used, at time $1+2\pi/3-\epsilon$, both robots searched (collectively) part of the circle $\mathcal{C}$ of measure at most $2 (2\pi/3-\epsilon) = 4\pi/3 - 2\epsilon$. It implies that part of $\mathcal{C}$ that is not visited yet has measure at least $2\pi/3 + 2\epsilon$. We will show that there exist two points at distance $\sqrt{3}$ that are still not visited. This will finish the proof of a lower bound. Indeed, we can wait for one of the robots to visit one of these points. When this happens, we can place the treasure in the other point forcing this robot to work for $\sqrt{3}$ extra time. It implies that the time to complete the task is at least $1+2\pi/3-\epsilon + \sqrt{3}$. The desired bound holds after taking $\epsilon \to 0$. 

It remains to prove the claim. For a contradiction, suppose that there exists a set $X \subseteq \mathcal{C}$ of measure $\mu(X) > 2\pi/3$ with the property that no two points in $X$ are at distance $\sqrt{3}$. Note that a distance of exactly $\sqrt{3}$ corresponds to an angular separation of exactly $2\pi/3$. Now, let us create two new sets by rotating $X$ around the circle: let $X_1 = X$ be the original set $X$, let $X_2$ be the set $X$ rotated by $2\pi/3$, and let $X_3$ be the set $X$ rotated by $4\pi/3$. Because, rotating a set does not change its size, all sets have the same measure. 

Consider an intersection of any pair of these three sets, $X_i \cap X_j$. If this intersection is non-empty, then there exists a point separated by $2\pi/3$ which we assumed is impossible. Hence, all three sets are mutually disjoint subsets of the circle $\mathcal{C}$. But it means that
$$
2\pi = \mu(\mathcal{C}) \ge \mu(X_1 \cup X_2 \cup X_3) = \mu(X_1) + \mu(X_2) + \mu(X_3) > 3 \mu(X) > 2 \pi,
$$
which gives us the desired contradiction. This finishes the proof of the theorem. 
\end{proof}

\subsection{Two Robots and Two Treasures ($n=2$, $k=2$)} \label{sec:n2k2}

We were lucky so far. We managed to determine the value of $c_k$ for any $k \in \N$. We also computed $\rho_k^1$ for any $k \in \N$ and $\rho_1^2$. The next step is to try to find the value of $\rho_2^2$. Unfortunately, our luck ends here. We only managed to prove the following bounds for $\rho_2^2$.

\begin{ptheorem}\label{thm:n2k2}
$5.5675 \approx 1+\frac {\pi}{3} + c_2 \le \rho_2^2 \le 6.2195$. \\
($6.2195$ is an upper bound for some implicitly defined constant $U$; see~\eqref{eq:def_U}.)
\end{ptheorem}

The lower bound is clearly stronger than the simple one of $1+c_2$ implied by \cref{lem:rho_lower}. In fact, we will prove a more general lower bound for $\rho_k^n$ (see \cref{thm:org_LB}, equation~\eqref{eq:org_LB2}, in \cref{sec:org_LB}). For the upper bound see \cref{thm:org_UB} in \cref{sec:org_UB}.

\subsubsection{Lower Bound} \label{sec:org_LB}

\begin{ptheorem}\label{thm:org_LB}
For any $n, k \in \N$, 
\begin{equation}\label{eq:org_LB1}
\rho_k^n \ge 1 + \frac {2\pi}{(k+1)n}+2k \sin \left( \frac {\pi}{k+1} \right).
\end{equation}
Moreover, for a special case ($k=2$) we get the following: for any $n \in \N$,
\begin{equation}\label{eq:org_LB2}
\rho_2^n \ge 1 + \frac {2\pi}{3n} + c_2 \approx 5.5675.
\end{equation}
\end{ptheorem}

\begin{proof}
Fix any $\epsilon > 0$, arbitrarily small. Let $S \in \mathcal{S}_k^n$ be an arbitrary strategy for $n$ robots. For a given $i \in [n]$, let $X_i$ be the part of the circle $\mathcal{C}$ searched by robot $i$ by time $1+\frac {2\pi}{(k+1)n}-\epsilon$. Trivially, regardless of the strategy $S$ used, each set $X_i$ has measure at most $\frac {2\pi}{(k+1)n}-\epsilon$ and so all robots searched (collectively) set $\bigcup_i X_i$, part of the circle $\mathcal{C}$ of measure at most $\frac {2\pi}{k+1}-\epsilon n$. 

For a given $\alpha \in [0, 2\pi)$, let $G_\alpha$ be a ``rooted'' regular $(k+1)$-gon inscribed in $\mathcal{C}$, that is, $G_\alpha$ has $k+1$ vertices: the $i$-th vertex $V_i = (\cos(\theta_i), \sin(\theta_i))$ is located at angle $\theta_i = \alpha + (i-1) \frac{2\pi}{k+1}$. Note that each $(k+1)$-gon $G_\alpha$ in this family has a unique root, vertex $V_1$, located at angle~$\alpha$. Our goal is to show that there exists a root $V \in \mathcal{C}$ (and the corresponding $\alpha \in [0, 2\pi)$) such that no vertex in $G_\alpha$ is visited by any of the $n$ robots. Let $Y \subseteq \mathcal{C}$ be the set of roots without this property. Since any point visited by one of the robots eliminates $k+1$ roots as potential candidates, we get that 
$$
\mu(Y) \le (k+1) \Big| \bigcup_i X_i \Big| \le (k+1) \left( \frac {2\pi}{k+1}-\epsilon n \right) = 2\pi - \epsilon n (k+1) < 2\pi.
$$
Hence, $\mu(\mathcal{C} \setminus Y) > 0$ so, indeed, there exists $G_\alpha$ with no vertex visited by any of the robots.

The conclusion now follows easily. We can wait for the first time when one of the robots (let us call it $r^*$) visits one of the vertices of $G_\alpha$. At this point we put $k$ treasures in the remaining $k$ vertices of $G_\alpha$. It is easy to see that the fastest way for robot $r^*$ to visit these treasures is to visit them in clockwise (or counter-clockwise) order; see the proof of \cref{lem:general_LB_ck}. Hence, it will take at least $k \cdot 2 \sin( \frac {\pi}{k+1} )$ units of time for $r^*$ to finish the job. We get the lower bound of $1+\frac {2\pi}{(k+1)n}-\epsilon+2k \sin( \frac {\pi}{k+1} )$ and \eqref{eq:org_LB1} holds after taking $\epsilon \to 0$.  

\medskip

Let us now deal with the special case ($k=2$) and the lower bound in \eqref{eq:org_LB2}. The first part of the proof works exactly as before. The only difference is that instead of $G_\alpha$ being a regular $3$-gon, we consider $3$-gon constructed by taking $2$ vertices of a $(2,\alpha)$-special placement of treasures together with a starting point $R=(1,0)$; see \cref{def:special_family}. We select an optimal value of $\alpha$ that yields $c_2$ (see \cref{sec:c2}) and we root the 3-gon at $R$. As before, we wait for the first robot to visit one of the vertices of the 3-gon and put the two treasures in the other two vertices. The 3-gon is not regular but, conveniently, even if non-root is visited first by some robot, the time to visit the remaining vertices is also $c_2$. This proves \eqref{eq:org_LB2} and the proof of the theorem is finished.
\end{proof}

Let us make a few remarks. First of all, note that the argument that gives a stronger bound for $k=2$ cannot be extended to $k \ge 4$. We are guaranteed to find the $(k+1)$-gone associated with the optimal placement of treasures yielding $c_k$, and we can wait for the first vertex of this $(k+1)$-gone to be visited. However, the time to visit the remaining $k$ vertices is not necessarily at least $c_k$ so the argument brakes down at this point. The 4-gon corresponding to the case $k=3$ is in fact regular so the alternative argument works but it is equivalent and gives the same bound as \eqref{eq:org_LB1}. 

Note also, as already mentioned earlier, that for any $n \in \N$ the lower bound in \eqref{eq:org_LB2} is stronger than the simple one of $1+c_2$ implied by \cref{lem:rho_lower}. The lower bound in \eqref{eq:org_LB1} is stronger for many values of $n$ and $k$ but clearly not for all of them. If $k$ is fixed, then \eqref{eq:org_LB1} is weaker than $1+c_k$ implied by \cref{lem:rho_lower}, provided that $n$ is large enough. 

\subsubsection{Upper Bound} \label{sec:org_UB}

Before we prove an upper bound, we need to introduce some definitions and do some calculus. Let 
$$
f(\beta) := 1 + 2\sin(\beta)+\frac{2\pi}{3} + \sqrt{3}, \qquad \beta \in (0, \pi/2).
$$
For a fixed $\beta \in (0, \pi/2)$, let 
$$
h_\beta(\gamma) := 1+\beta+\gamma+2\sin(\beta+\gamma)+2\sin(\gamma/2), \qquad \gamma \in (0, \pi-\beta).
$$
Clearly,
\begin{eqnarray*}
h_\beta'(\gamma) &=& 1+2\cos(\beta+\gamma)+\cos(\gamma/2) \\
h_\beta''(\gamma) &=& -2\sin(\beta+\gamma) - \frac 12 \sin(\gamma/2).
\end{eqnarray*}
Since $\beta+\gamma \in (0,\pi)$ and $\gamma/2 \in (0, \pi/2)$, $h_\beta''(\gamma)$ is strictly negative and so $h_\beta(\gamma)$ is strictly concave, guaranteeing that a critical point where $h_\beta'(\gamma) = 0$ is the unique global maximum on the interval $(0, \pi-\beta)$. Let $\gamma^* = \gamma^*(\beta)$ be the critical point that maximizes function $h_\beta(\gamma)$ and let $g(\beta)$ be the corresponding maximum value, that is, let
$$
g(\beta) = \max_{\gamma \in (0, \pi-\beta)} h_\beta(\gamma) = h_\beta(\gamma^*).
$$

Function $f(\beta)$ is clearly an increasing function of $\beta$. Function $g(\beta)$, on the other hand, is a decreasing function of $\beta$ but this is far from being obvious. For this we will use the multivariable chain rule. Note that 
$$
g'(\beta) = \frac {d h_\beta(\gamma^*)}{d\beta} = \frac {\partial h_\beta(\gamma^*)}{\partial \beta} + \frac {\partial h_\beta(\gamma^*)}{\partial \gamma} \frac {d \gamma^*}{d \beta} = \frac {\partial h_\beta(\gamma^*)}{\partial \beta},
$$
since we are evaluating at the critical point $\gamma^*$ and so the term $\frac{\partial h_\beta(\gamma^*)}{\partial \gamma} = 0$, causing the second term to vanish. We get that 
\begin{eqnarray*}
g'(\beta) &=& \frac {\partial}{\partial \beta} \Big( 1+\beta+\gamma^*+2\sin(\beta+\gamma^*)+2\sin(\gamma^*/2) \Big) \\
&=& 1 + 2\cos(\beta+\gamma^*).
\end{eqnarray*}
Since $\gamma^*$ is a critical point, we get that 
$$
h_\beta'(\gamma^*) = 1+2\cos(\beta+\gamma^*)+\cos(\gamma^*/2) = 0.
$$
We conclude that $g'(\beta) = - \cos(\gamma^*/2) < 0$, since $\gamma^*/2 \in (0, \pi/2)$, and so indeed $g(\beta)$ is a decreasing function of $\beta$. 

These observations allow us to define the following constants. Let
\begin{equation}\label{eq:def_U}
U := \min_{\beta \in (0, \pi/2)} \max \Big\{ f(\beta), g(\beta) \Big\} = f(\beta^*) = g(\beta^*),
\end{equation}
where $\beta^*$ is the unique value in $(0, \pi/2)$ such that $f(\beta) = g(\beta)$. There is no closed formula for $U$ but one can easily approximate it numerically. We report that $U < 6.2195$. This value was obtained using the code that can be found in Appendix \ref{sec:appendix_code_rho22}.

Now, we are ready to prove an upper bound.

\begin{ptheorem}\label{thm:org_UB}
$\rho_2^2 \le U < 6.2195$.
\end{ptheorem}

\begin{proof}
Let $\beta^*$ be the unique value of $\beta$ in $(0, \pi/2)$ such that $U = f(\beta) = g(\beta)$ (see the discussion above). We will analyze the following strategy. Both robots go straight to point $(-1,0)$ and then they split, one of them goes along the circle $\mathcal{C}$ clockwise whereas the other one goes along the circle $\mathcal{C}$ counter-clockwise. We distinguish two cases. 

\medskip
\textbf{Case 1:} One of the treasures is discovered at time $1+\beta$ for some $\beta < \beta^*$. By symmetry, we may assume that it is at point $X = (\cos(\pi-\beta), \sin(\pi-\beta))$ and so it is discovered by the first robot. At that point, the second robot is at point $X' = (\cos(\pi-\beta), -\sin(\pi-\beta))$, a reflection of $X$ across the $x$-axis. The first robot runs from $X$ straight to $X'$ whereas the second robot runs from $X'$ to $X$. They reach their destinations after additional $2 \sin(\beta)$ units of time. 

If the second treasure was placed in $X'$ (that is, both treasures were discovered at the same time), then the job is already done at this point. But, since we search for the worst possible placement of treasures, we may assume that at this point only one treasure is discovered and visited by both robots. The robots continue traversing the circle $\mathcal{C}$. Once the second treasure is discovered, one of the robots is already done and the other one has to run straight to the missing treasure. Our goal now is to find the worst position of the second treasure.

Suppose that the second treasure is found at point $Y = (\cos(\pi-\xi), \sin(\pi-\xi))$. Clearly, $\xi < \pi/2$ does not correspond to the worst-case scenario as $\xi = \pi/2$ yields a worse placement: it takes more time to discover the second treasure and at this point the distance between $Y$ and its reflection, $Y'$, is equal to 2 (the maximum distance two points on the circle can be away from each other). Hence, we may assume that $\xi \ge \pi/2$. But then, it follows from \cref{lem:sqrt3jump} that the worst case scenario for placing the second treasure is to place it at angle $\xi = \pi/2 + \pi/6$ so that the distance between $Y$ and $Y'$ is equal to $\sqrt{3}$. We conclude that the time for the robots to finish their task (in this case) is at most
$$
(1 + \beta) + 2 \sin( \beta ) + \left( \frac {\pi}{2} - \beta \right) + \frac {\pi}{6} + \sqrt{3} ~\le~ 1 + 2 \sin ( \beta^* ) + \frac {2\pi}{3} + \sqrt{3} ~=~ f(\beta^*).
$$

\medskip
\textbf{Case 2:} One of the treasures is discovered at time $1+\beta$ for some $\beta \ge \beta^*$. As in the previous case, by symmetry we may assume that the first treasure is located at point $X = (\cos(\pi-\beta), \sin(\pi-\beta))$. The robots continue traversing the circle until the second treasure, located at point $Y$, is discovered. When this happens, the robots visit missing treasures (if there are any) as quickly as possible. 

It is clear that the worst case scenario is when the first robots discovers both treasures. The second robot goes straight to $Y$ and then goes straight to $X$. It is easy to see that the worst case placement for $X$ is to place it at $X = (\cos(\pi-\beta^*), \sin(\pi-\beta^*))$ since this choice maximizes the distance between $Y$ and $X$, given the constraint that $\beta \ge \beta^*$. Assuming that $Y = (\cos(\pi-\beta^*-\gamma), \sin(\pi-\beta^*-\gamma))$, we conclude that the time for the robots to finish their task (in this case) is equal to
$$
1 + \beta^* + \gamma + 2 \sin(\beta^*+\gamma) + 2 \sin(\gamma/2) ~=~ h_{\beta^*}(\gamma) ~\le~ h_{\beta^*}(\gamma^*) ~=~ g(\beta^*).
$$

Hence, regardless of where the first treasure is placed (Case 1 or Case 2), the strategy guarantees that both treasures are visited by both robots before time $U = f(\beta^*) = g(\beta^*)$. This finishes the proof of the theorem. 
\end{proof}

\newpage

\appendix

\section{Table of Approximated Values of $s_k$}\label{sec:appendix_tables_sk}

\begin{table}[h!]
{
\centering
\begin{tabular}{|r|r|c|c|l|}
\hline
$k$ & optimal $\alpha$ & $s_k$  & $j$ & $\beta_\ell$ for $\ell\in[j+1]$ (in degrees)\\
\hline
2   & $107.2496^{\circ}$ & 3.5203 & 1 & $(107.25, 145.50)$                                     \\
3   &  $60.0000^{\circ}$ & 4.4641 & 1 & $( 60.00, 120.00)$                                     \\
4   &  $43.2721^{\circ}$ & 5.0155 & 2 & $( 43.27,  86.54, 100.37)$                             \\
5   &  $31.2740^{\circ}$ & 5.3481 & 2 & $( 31.27,  62.55,  86.18)$                             \\
6   &  $21.5917^{\circ}$ & 5.5741 & 2 & $( 21.59,  43.18,  76.82)$                             \\
7   &  $16.0204^{\circ}$ & 5.7373 & 3 & $( 16.02,  32.04,  64.08, 67.86)$                      \\
8   &  $12.8023^{\circ}$ & 5.8496 & 3 & $( 12.80,  25.60,  51.21, 60.26)$                      \\
9   &  $10.1139^{\circ}$ & 5.9310 & 3 & $( 10.11,  20.23,  40.46, 54.60)$                      \\
10  &   $7.7389^{\circ}$ & 5.9934 & 3 & $(  7.74,  15.48,  30.96, 50.33)$                      \\
11  &   $5.7666^{\circ}$ & 6.0435 & 4 & $(  5.77,  11.53,  23.07, 46.13, 46.75)$               \\
12  &   $4.8899^{\circ}$ & 6.0823 & 4 & $(  4.89,   9.78,  19.56, 39.12, 42.66)$               \\
13  &   $4.1262^{\circ}$ & 6.1124 & 4 & $(  4.13,   8.25,  16.50, 33.01, 39.37)$               \\
14  &   $3.4397^{\circ}$ & 6.1365 & 4 & $(  3.44,   6.88,  13.76, 27.52, 36.69)$               \\
15  &   $2.8027^{\circ}$ & 6.1562 & 4 & $(  2.80,   5.61,  11.21, 22.42, 34.49)$               \\
16  &   $2.1902^{\circ}$ & 6.1728 & 4 & $(  2.19,   4.38,   8.76, 17.52, 32.70)$               \\
17  &   $1.8469^{\circ}$ & 6.1869 & 5 & $(  1.85,   3.69,   7.39, 14.78, 29.55, 30.69)$        \\
18  &   $1.6303^{\circ}$ & 6.1984 & 5 & $(  1.63,   3.26,   6.52, 13.04, 26.08, 28.77)$        \\
19  &   $1.4318^{\circ}$ & 6.2080 & 5 & $(  1.43,   2.86,   5.73, 11.45, 22.91, 27.12)$        \\
20  &   $1.2467^{\circ}$ & 6.2162 & 5 & $(  1.25,   2.49,   4.99,  9.97, 19.95, 25.70)$        \\
21  &   $1.0708^{\circ}$ & 6.2231 & 5 & $(  1.07,   2.14,   4.28,  8.57, 17.13, 24.47)$        \\
22  &   $0.8998^{\circ}$ & 6.2291 & 5 & $(  0.90,   1.80,   3.60,  7.20, 14.40, 23.40)$        \\
23  &   $0.7295^{\circ}$ & 6.2345 & 5 & $(  0.73,   1.46,   2.92,  5.84, 11.67, 22.48)$        \\
24  &   $0.6478^{\circ}$ & 6.2391 & 6 & $(  0.65,   1.30,   2.59,  5.18, 10.36, 20.73, 21.41)$ \\
\hline
\end{tabular}
\caption{Approximated values of $s_k$ and corresponding optimal $(k,\alpha)$-special placements.}
\label{table:optimal}
}
\end{table}

\newpage
\section{Table of Approximated Values or Bounds for $\rho^n_k$}\label{sec:appendix_tables_rho}

\begin{table}[h!]
\begin{tabular}{|c|ccccc|}
\hline
$k ~\backslash~ n$ & 1         & 2                & 3                & 4                & 5                \\
\hline
1         & 7.2831    & 4.8264           & (4.0472, 5.0944) & (3.7854, 4.5708) & (3.6283, 4.2566) \\
2         & 7.2831    & (5.5675, 6.2195) & (5.1622, 6.6147) & (4.9877, 6.0911) & (4.8829, 5.7769) \\
3         & 7.2831    & (6.0280, 7.2831) & (5.7662, 7.2831) & (5.6353, 7.0349) & (5.5568, 6.7207) \\
4         & 7.2831    & (6.3306, 7.2831) & (6.1211, 7.2831) & (6.0164, 7.2831) & (6.0155, 7.2721) \\
5         & 7.2831    & (6.5236, 7.2831) & (6.3490, 7.2831) & (6.3481, 7.2831) & (6.3481, 7.2831) \\
6         & 7.2831    & (6.6554, 7.2831) & (6.5741, 7.2831) & (6.5741, 7.2831) & (6.5741, 7.2831) \\
7         & 7.2831    & (6.7503, 7.2831) & (6.7373, 7.2831) & (6.7373, 7.2831) & (6.7373, 7.2831) \\
8         & 7.2831    & (6.8496, 7.2831) & (6.8496, 7.2831) & (6.8496, 7.2831) & (6.8496, 7.2831) \\
9         & 7.2831    & (6.9310, 7.2831) & (6.9310, 7.2831) & (6.9310, 7.2831) & (6.9310, 7.2831) \\
10        & 7.2831    & (6.9934, 7.2831) & (6.9934, 7.2831) & (6.9934, 7.2831) & (6.9934, 7.2831) \\
11        & 7.2831    & (7.0435, 7.2831) & (7.0435, 7.2831) & (7.0435, 7.2831) & (7.0435, 7.2831) \\
12        & 7.2831    & (7.0823, 7.2831) & (7.0823, 7.2831) & (7.0823, 7.2831) & (7.0823, 7.2831) \\
13        & 7.2831    & (7.1124, 7.2831) & (7.1124, 7.2831) & (7.1124, 7.2831) & (7.1124, 7.2831) \\
14        & 7.2831    & (7.1365, 7.2831) & (7.1365, 7.2831) & (7.1365, 7.2831) & (7.1365, 7.2831) \\
15        & 7.2831    & (7.1562, 7.2831) & (7.1562, 7.2831) & (7.1562, 7.2831) & (7.1562, 7.2831) \\
16        & 7.2831    & (7.1728, 7.2831) & (7.1728, 7.2831) & (7.1728, 7.2831) & (7.1728, 7.2831) \\
17        & 7.2831    & (7.1869, 7.2831) & (7.1869, 7.2831) & (7.1869, 7.2831) & (7.1869, 7.2831) \\
18        & 7.2831    & (7.1984, 7.2831) & (7.1984, 7.2831) & (7.1984, 7.2831) & (7.1984, 7.2831) \\
19        & 7.2831    & (7.2080, 7.2831) & (7.2080, 7.2831) & (7.2080, 7.2831) & (7.2080, 7.2831) \\
20        & 7.2831    & (7.2162, 7.2831) & (7.2162, 7.2831) & (7.2162, 7.2831) & (7.2162, 7.2831) \\
21        & 7.2831    & (7.2231, 7.2831) & (7.2231, 7.2831) & (7.2231, 7.2831) & (7.2231, 7.2831) \\
22        & 7.2831    & (7.2291, 7.2831) & (7.2291, 7.2831) & (7.2291, 7.2831) & (7.2291, 7.2831) \\
23        & 7.2831    & (7.2345, 7.2831) & (7.2345, 7.2831) & (7.2345, 7.2831) & (7.2345, 7.2831) \\
24        & 7.2831    & (7.2391, 7.2831) & (7.2391, 7.2831) & (7.2391, 7.2831) & (7.2391, 7.2831) \\
\hline
\end{tabular}
\caption{Approximated values or bounds for $\rho^n_k$. When an exact value is known, its approximation is reported; otherwise, lower and upper bounds are given.}
\label{table:bounds}
\end{table}

\section{Code Used to Identify the Worst-case Placement of Treasures}\label{sec:appendix_code_special_family}

\begin{verbatim}
using Combinatorics
using Optimization
using OptimizationMultistartOptimization
using OptimizationOptimJL
using ForwardDiff

d(x, y) = sqrt((x[1] - y[1])^2 + (x[2] - y[2])^2)

genpermutations(n) = collect(permutations(1:n))

const perms = Dict{Int, Vector{Vector{Int64}}}()
const MAX_SPOTS = 8

for i in 1:MAX_SPOTS
    perms[i] = genpermutations(i)
end

function C(T)
    best = Inf
    for p in perms[length(T)]
        cur = d((1.0, 0.0), T[p[1]])
        for i in 2:length(T)
            cur += d(T[p[i-1]], T[p[i]])
        end
        if cur < best
            best = cur
        end
    end
    return -best
end

CA(A, p) = C([(cos(a), sin(a)) for a in A]) # version with angles 

function run_opt(n, reps)
    lb = fill(0.0, n)
    ub = fill(2 * pi, n)
    x0 = rand(n) * pi

    optf = OptimizationFunction(CA, Optimization.AutoForwardDiff())
    prob = OptimizationProblem(optf, x0; lb = lb, ub = ub)
    sol = solve(prob, MultistartOptimization.TikTak(reps), LBFGS())
    return sol
end

for i in 1:MAX_SPOTS
    @show i
    s = [run_opt(i, 100) for _ in 1:10]
    best = s[1]
    for j in 2:length(s)
        if best.objective > s[j].objective
            best = s[j]
        end
    end
    @show best
end
\end{verbatim}

\section{Code Used to Approximate $s_k$}\label{sec:appendix_code_sk}

\begin{verbatim}
function c(x, k)
    maxj = div(k, 2)
    last = 2*pi
    for j = 1:maxj
        current = 4.0*sum(sin(x*(2.0^(a-2.0))) for a in 1:j) + 
            2.0*(k+1.0-2.0*j)*sin((pi-(2.0^j-1.0)*x)/(k+1.0-2.0*j))-
            2.0*sin(0.5*x)
        if current < last
            last = current
        else
            throw(ErrorException("non decreasing sequence"))
        end
        if pi / ((2.0^(j-1.0))*(k+3.0-2.0*j)-1.0) <=
               x <=
               pi / ((2.0^(j-2.0))*(k+5.0-2.0*j)-1.0)
            return current, j
        end
    end
    throw(ErrorException("something went wrong"))
end

function optc(k)
    println("\nk:\t$k")
    minj = 1
    maxj = div(k, 2)
    minx = pi/((2.0^(maxj-1))*(k+3-2*maxj)-1.0)
    maxx = pi/((2.0^(minj-2))*(k+5-2*minj)-1.0)
    x = range(minx, maxx, length=100000)
    y = c.(x, k)
    sk, j = maximum(y)
    bx = x[argmax(y)]*180/pi
    println("s(k):\t", sk)
    println("betas:\t", push!([2*(bx*(2.0^(a-2.0))) for a in 1:j],
                              (360-2(2.0^j-1.0)*bx)/(k+1.0-2.0*j)))
end

for k=2:24
    optc(k)
end
\end{verbatim}

\section{Code Used to Approximate the Upper bound for $\rho_2^2$}\label{sec:appendix_code_rho22}

\begin{verbatim}
using Optim
f(b) = 1 + 2 * sin(b) + 2 * pi / 3 + sqrt(3)
h(b, c) = 1 + b + c + 2 * sin(b + c) + 2 * sin(c / 2)
g(b) = -Optim.minimum(optimize(c -> -h(b, c), 0, pi - b))
U = Optim.minimum(optimize(b -> max(f(b), g(b)), 0.0, pi / 2))
\end{verbatim}

\end{document}